%% file: ref.tex
\documentclass[final]{article}

\usepackage{xcolor}
\usepackage{amsthm}
\usepackage{enumitem}
\usepackage{soul}
\usepackage{graphicx}
\usepackage[linesnumbered,ruled,vlined,resetcount]{algorithm2e}

\include{macro_AP}

\begin{document}
	
	\title{First Approximation for Uniform Lower and Upper Bounded Facility Location Problem avoiding violation in Lower Bounds}
	
	\maketitle
	\begin{center}
		\author{Sapna Grover$^1$,}
		\author{Neelima Gupta$^2$ and}
		\author{Rajni Dabas$^3$}
	\end{center}
	\begin{enumerate}
		\item {Department of Computer Science, University of Delhi, India.\\
			\texttt{sgrover@cs.du.ac.in, sapna.grover5@gmail.com}}
		\item {Department of Computer Science, University of Delhi, India.\\
			\texttt{ngupta@cs.du.ac.in}}
		\item {Department of Computer Science, University of Delhi, India.\\
			\texttt{rajni@cs.du.ac.in}}
	\end{enumerate}
	
	\input{abstract}


\input{Intro_rd}
\input{TriCriteria_l_1Minus1Byl}

\input{LBUBFLGeneralTree}
\input{Conclusion}

\bibliography{ref}

\end{document}

%% file: macro_AP.tex
\newcommand{\facilityset}{\mathcal{F}}	
	
\newcommand{\clientset}{\mathcal{C}}			
\newcommand{\cliset}{\mathcal{C}}				

\newcommand{\capacityU}{\mathcal{U}}				

\newcommand{\cdense}{\mathcal{C}_D}				
\newcommand{\csparse}{\mathcal{C}_S}

\newcommand{\ballofj}[1]{\textit{$\mathcal{B}( #1 $)}} 	

\newcommand{\C}[1]{\hat{C_{#1}}}					

\newcommand{\bundle}[1]{\mathcal{N}_{#1}}			

\newcommand{\dist}[2]{c(#1,~#2)}

\newcommand{\facilitycost}{\textit{$ f_i $}}		

\newcommand{\singleclient}{\textit{j}}
\newcommand{\singlefacility}{\textit{i}}

\newcommand{\zofi}[1]{z_#1}

\newcommand{\primezofi}[1]{z'_{#1}}

\newcommand{\demandofj}[1]{\textit{$ d_{#1} $}}

\newcommand{\loadjinc}[2]{\phi(#1,~#2)}

\newcommand{\cardTwo}[1]{Y^*(#1)}

\newcommand{\etal}{\textit{et al}.}

\newcommand{\B}{\mathcal{L}}

\newcommand{\opt}[1]{LP_{opt}}

\newtheorem{theorem}{Theorem}[section]

\newtheorem{lemma}[theorem]{Lemma}
\newtheorem{claim}[theorem]{Claim}
\newtheorem{observation}[theorem]{Observation}
\newtheorem{definition}[theorem]{Definition}
\newcommand{\rhobar}{\bar\rho}
\newcommand{\cone}{y}

\newcommand{\ellone}{\ell}

\newcommand{\rootpair}{<r_1, r_2>}

%% file: abstract.tex
\begin{abstract}
With growing emphasis on e-commerce marketplace platforms where we have a central platform mediating between the seller and the buyer, it becomes important to keep a check on the availability and profitability of the central store. A store serving too less clients can be non-profitable and a store getting too many orders can lead to bad service to the customers which can be detrimental for the business.
In this paper, we study the facility location problem(FL) with upper and lower bounds on the number of clients an open facility serves. 
Constant factor approximations are known for the restricted variants of the problem with only the upper bounds or only the lower bounds. The only work that deals with bounds on both the sides violates both the bounds \cite{friggstad_et_al_LBUFL}.
In this paper, we present the first (constant factor) approximation for the problem violating the upper bound by a factor of 
$(5/2)$ without violating the lower bounds when both the lower and the upper bounds are uniform.  
We first give a tri-criteria (constant factor) approximation violating both the upper and the lower bounds and 
then get rid of violation in lower bounds by transforming the problem instance to an instance of capacitated facility location problem. 

\end{abstract}

%% file: Intro_rd.tex
\section{Introduction}
Facility location problem (FL) is a well motivated and extensively studied problem. Given a set of facilities with facility opening costs and a set of clients with a metric specifying the connection costs between facilities and clients, the goal is to select a subset of facilities such that the total cost of opening the selected facilities and connecting clients to the opened facilities is minimized.

With growing emphasis on e-commerce marketplace platforms where we have a central platform mediating between the seller and the buyer, it becomes important to keep a check on the availability and profitability of the central store. A store serving too less clients can be non-profitable and a store getting too many orders can lead to bad service to the customers. This scenario leads to what we call as the lower- and upper- bounded facility location (LBUBFL) problem.
There are several other applications requiring both the lower as well as the upper bound on the number of clients assigned to the selected facilities. In a real world transportation problem presented by Lim et al.~\cite{Lim_transportation}, there are a set of carriers and a set of customers who wish to ship a certain number of cargoes through these carriers. There is a transportation cost associated to ship a cargo through a particular carrier. The goal of the company is to ship cargoes in these carriers such that the total transportation cost is minimized. There is a natural upper bound on the number of cargoes a carrier can carry imposed by the capacity of the carrier.  In addition, the shipping companies are required to engage their carriers with a \textquotedblleft minimum quantity commitment\textquotedblright~when shipping cargoes to United States. 
Another scenario in which the problem can be useful is in the applications requiring balancing of load 
on the facilities along with maintaining capacity constraints.

In this paper, we study the facility location problem with lower and upper bounds. We are given a set $ \clientset$ of clients and a set $\facilityset$ of facilities with lower bounds $\B_i$ and upper bounds $\capacityU_i$ on the minimum and the maximum number of clients a facility $i$ can serve, respectively. Setting up a facility at location $i$ incurs cost $f_i$(called the {\em facility opening cost}) and servicing a client $j$ by a facility $i$ incurs cost $\dist{i}{j}$ (called the {\em service cost}). We assume that the costs are metric, i.e., they satisfy the triangle inequality. Our goal is to open a subset $\facilityset' \subseteq \facilityset$ and compute an assignment function $\sigma : \clientset \rightarrow \facilityset'  $ (where $\sigma(j)$ denotes the facility that serves $j$ in the solution)
such that $ \B_i \leq |\sigma^{-1}(i)| \leq  \capacityU_i ~\forall i \in \facilityset'$
and, the total cost of setting up the facilities and servicing the clients is minimised. The problem is known to be NP-Hard.  We present the first (constant factor) approximation for the problem with uniform lower and uniform upper bounds, i.e., $\B_i=\B$ and $\capacityU_i=\capacityU~\forall i \in \facilityset$ without violating the lower bounds, as stated in Theorem \ref{mainThm}.

\begin{definition}
A tri-criteria $ (\alpha, \beta, \gamma)$- approximation for LBUBFL problem is a solution $S= (\facilityset',\sigma$)
satisfying $\alpha \B \leq~|\sigma^{-1}(i)| \leq \beta \capacityU ~\forall i \in \facilityset' , \alpha \le 1, \beta \ge 1$, with cost no more than  $\gamma OPT$, where $OPT$ denotes the cost of an optimal solution of the problem. 
\end{definition}

\begin{theorem} \label{mainThm}
 A $(1, 5/2, O(1))$- 
 approximation can be obtained for LBUBFL in polynomial time.
\end{theorem}

Constant factor approximations are known for the problem with upper bounds only (popularly known as Capacitated Facility Location (CFL)) with~\cite{Shmoys,capkmByrkaFRS2013,GroverGKP18} and without~\cite{KPR,ChudakW99,mathp,paltree,mahdian_universal,zhangchenye,Bansal,Anfocs2014} violating the capacities using local search / LP rounding techniques. Constant factor approximations are also known for the problem with lower bounds only with~\cite{Minkoff_LBFL,Guha_LBFL,Han_LBkM} and without~\cite{Zoya_LBFL,Ahmadian_LBFL,Li_NonUnifLBFL} violating the lower bounds. The only work that deals with the bounds on both the sides is due to Friggstad et al. \cite{friggstad_et_al_LBUFL},
which deals with the problem with non-uniform lower bounds and uniform upper bounds.  They gave a constant factor approximation for the problem using LP-rounding,  violating both the upper and the lower bound by a constant factor. The technique cannot be used to get rid of the violation in the lower bounds even if they are uniform as the authors show an unbounded integrality gap for the problem. Thus, our result is an improvement over them when the lower bounds are uniform in the sense that they violate both the bounds whereas we do not violate the lower bounds.


\subsection{Related Work}
For capacitated facility location with  uniform capacities, Shmoys et al.~\cite{Shmoys} gave the first constant factor($7$) algorithm with a capacity blow-up of $7/2$ using LP rounding techniques.   An $O(1/\epsilon^2)$ factor approximation, with $(2 + \epsilon)$ violation in capacities, follows as a special case of C$k$FLP by Byrka et al.~\cite{capkmByrkaFRS2013}. Grover et al.~\cite{GroverGKP18} reduced the capacity violation to $(1 + \epsilon)$. For non-uniform capacities, Levi et al. \cite{LeviSS12} gave a 5-factor approximation algorithm using LP-rounding for a restricted version of the problem in which the facility opening costs are uniform. Later, An et al.~\cite{Anfocs2014} gave the first LP-based constant factor approximation algorithm without violating the capacities, by strengthening the natural LP. 
 
 The local search technique has been particularly useful to deal with capacities. Korupolu et al.~\cite{KPR} gave the first constant factor ($8+\epsilon$) approximation for  uniform capacities, which was further improved by Chudak and Williamson  \cite{ChudakW99} to ($6+\epsilon$). 
 The factor was subsequently reduced to $3$ by Aggarwal et al.~\cite{mathp}, which is also the best known result for the problem. For non-uniform capacities, Pal et al.~\cite{paltree} gave the first constant factor $(8.53 + \epsilon)$ approximation algorithm which was subsequently improved to $(7.88 +\epsilon)$ by Mahdian and Pal~\cite{mahdian_universal} and to  $(5.83+\epsilon)$ by Zhang et al.~\cite{zhangchenye} with the current best being $(5 + \epsilon)$ due to Bansal et al.~\cite{Bansal}. 

Lower-Bounded Facility Location (LBFL) problem was introduced by Karger and Minkoff \cite{Minkoff_LBFL} and Guha et al.~\cite{Guha_LBFL} independently in 2000. Both of them gave a bi-criteria algorithm, that is a constant-factor approximation with constant factor violation in lower bounds. Zoya Svitkina \cite{Zoya_LBFL} presented the first true constant factor($448$) approximation for the problem with uniform lower bounds by reducing the problem to CFL. This was later improved to $82.6$ by Ahmadian and Swamy \cite{Ahmadian_LBFL} using reduction to a special case of CFL called  Capacitated Discounted Facility Location. Later Shi Li \cite{Li_NonUnifLBFL} gave the first true constant ($4000$) factor approximation for the problem with non uniform lower bounds. Li obtained the result by reducing the problem  to CFL via two intermediate problems called, LBFL-P(Lower bounded facility location with penalties) and TCSD(Transportation with configurable supplies and demands).
Han et al. \cite{Han_LBkM} gave a bi-criteria solution for the problem as a particular case of lower bounded $k$-FL problem when the lower bounds are non-uniform.

Another variant of LBFL, called Lower bounded $k$-Median problem(LB$k$M) was considered by Guo et al. \cite{Guo_LBkM} and Arutyunova and Schmidt \cite{arutyunova_LBkM} with uniform lower bounds, where they gave constant factor approximations for the problems by reducing them to CFL and LBFL respectively. For Lower bounded $k$-Facility Location problem LB$k$FL problem with non-uniform lower bounds, Han et al. \cite{Han_LBkM} presented a bi-criteria algorithm giving $\alpha$ factor violaton in lower bounds and $\frac{1+\alpha}{1-\alpha}\rho$-approximation where $\rho$ is the approximation ratio for $k$FL problem and $\alpha$ is a constant. They also extend these results to LB$k$M and lower bounded knapsack median problem.

Friggstad et al.~\cite{friggstad_et_al_LBUFL} is the only work that deals with the problem with bounds on both the sides. The paper considers non-uniform lower bounds and uniform upper bounds. They gave a constant factor approximation, using LP rounding techniques, violating both the lower and the upper bounds. 

\subsection{High Level Idea}
Let $I$ be an input instance of the LBUBFL problem. We first present a tri-criteria solution $(> 1/2, 3/2, O(1))$ violating both the lower as well as the upper bound and then 
get rid of the violation in the lower bound by reducing the problem instance $I$ to an instance $I_{cap}$ of CFL via a series of reductions ($I \rightarrow I_1 \rightarrow I_2 \rightarrow I_{cap}$).
As we will see that maintaining  $\alpha > 1/2$ is crucial in getting rid of the violation in the lower bound and hence the tri-criteria solution of Friggstad et. al.~\cite{friggstad_et_al_LBUFL} cannot be used here as $\alpha < 1/2$ in their work. Also, when the lower bounds are uniform, our approach is comparatively simpler and straightforward.
%
Let $Cost_\mathcal{I}(\mathcal{S})$ denotes the cost of a solution $\mathcal{S}$ to an instance $\mathcal{I}$. We work in the following steps:
\begin{enumerate}
    \item 

    We first compute a tri-criteria solution 
    $(> 1/2, 3/2, O(1))$ - approximation
    $S^t = (\facilityset^t, \sigma^t $) to $I$ via clustering and filtering techniques. 
    Thus, $\alpha >1/2$ and $\beta = 3/2$.
    
   \item We transform the instance $I$ into another instance $I_1$ of LBUBFL by moving the clients assigned to a facility $i$ in the  tri-criteria solution $S^t$ to $i$.  Let $O$ denote the optimal solution of $I$.
   To construct a solution of $I_1$, a client $j$ is assigned to a facility $i$ if it was assigned to $i$ in $O$. The connection cost is bounded by the sum of the cost $j$ pays in $S^t$ and in $O$. Thus, $Cost_{I_1}(O_1) \le  Cost_I(S^t) + Cost_I(O)$ where $O_1$ denotes the optimal solution of $I_1$.

    \item The instance $I_1$ is transformed into another instance $I_2$ of LBFL, ignoring the upper bounds. The facility set is reduced to $\facilityset^t$ and the facilities in $\facilityset^t$ are assumed to be available free of cost. If a client $j$ is assigned, in $O_1$, to a facility $i$ not in $\facilityset^t$ then it is assigned to the facility $i'\in \facilityset^t$ nearest to $i$ (assume that the distances are distinct) and we open $i'$. The connection cost is bounded by twice the cost $j$ pays in $O_1$ by a simple triangle inequality. Let $O_2$ denote the optimal solution of $I_2$. Then,  $Cost_{I_2}(O_2) \le 2Cost_{I_1}(O_1)$.

    \item Finally, we create an instance $I_{cap}$ of CFL from $I_2$. The key idea in the reduction is: let $n_i$ be the number of clients assigned to a facility $i$ in $\facilityset^t$. If $i$ violates the lower bound, we create a demand of $\B - n_i$ at $i$ in the CFL instance otherwise we create a supply of $n_i - \B$ at $i$. The solution of the CFL instance then guides us to increase the assignments at some of the violating facilities until it gets $\B$ clients or decides to shut them.
    The process results in an increased violation in the capacities by plus $1$. 
    
    Let $O_{cap}$ be an optimal solution of $I_{cap}$. We show that  $Cost_{I_{cap}}(O_{cap}) \le (1+2\delta) Cost_{I_2}(O_2)$ for an appropriately chosen constant $\delta$.
    
    \item We obtain an $(5 + \epsilon)$-approximate solution $ AS_{cap}$ for $I_{cap}$ by using the algorithm of Bansal \etal~\cite{Bansal} for CFL.
    Thus, $Cost_{I_{cap}}(AS_{cap}) \le (5 + \epsilon) Cost_{I_{cap}}(O_{cap})$.
    
    
     \item From $ AS_{cap}$, we obtain an approximate solution $AS_1$ to $I_1$ 
     such that the upper bounds are violated by a factor of $(\beta + 1)$ with no violation in lower bounds and $Cost_{I_1}(S_1) \le \frac{2(\alpha+1)}{(2\alpha-1)} Cost_{I_{cap}}(AS_{cap})$. Facility trees are constructed and processed bottom-up. Clients are either moved up in the tree to the parent or to a sibling until we collect  at least $\B$ clients at a facility. Whenever $\B$ clients are assigned to a facility $i$, $i$ is opened and the subtree rooted at $i$ is chopped off the tree and the process is repeated with the remaining tree. 

   \item From $AS_1$, we obtain a solution $S$ to $I$ by paying the facility cost and moving the clients back to their original location. Thus, $ Cost_I(S) \le  Cost_I(S^t) + Cost_{I_1}(AS_1)$.

\end{enumerate}

Our main contributions are in Steps $1$ and $6$. 

\subsection{Organisation of the Paper}
In Section \ref{tri-criteria}, we present a tri-criteria algorithm for LBUBFL using LP rounding techniques. In Section \ref{instI1}, we reduce instance $I$ to $I_1$ and then $I_1$ to $I_2$, followed by reduction to  $I_{cap}$ in Section \ref{instIcap}. Finally, a bi-criteria solution, that does not violate the lower bounds, is obtained in Section \ref{bi-criteria}.

%% file: TriCriteria_l_1Minus1Byl.tex
\section{Computing the Tri-criteria Solution}
\label{tri-criteria}

 In this section, we first give a tri-criteria solution that violates the lower bound by a factor of $\alpha = (1 - \frac{1}{\ell})$ 
 and the upper bound by a factor of $\beta = (2 - \frac{1}{\ell})$, where $\ell \ge 2$ is a tune-able parameter. 
 This is one of the two major contributions of our work.
 In particular, we present the following results:

 \begin{theorem} \label{TricriteriaResult}
 An $((1 - \frac{1}{\ell}),(2 - \frac{1}{\ell}),(10\ell+4))$-
 approximate solution can be obtained for LBUBFL in polynomial time, where $\ell \ge 2$ is a tuneable parameter.
\end{theorem}
 
 
 Instance $I$ of LBUBFL can be formulated as the following integer program (IP):

\label{{unif-CFLP}}
$Minimize ~\mathcal{C}ostLBUBFL(x,y) = \sum_{j \in \clientset}\sum_{i \in \facilityset}\dist{i}{j}x_{ij} + \sum_{i \in \facilityset}f_iy_i $
\begin{eqnarray}
subject~ to &\sum_{i \in \facilityset}{} x_{ij} \geq 1 & \forall ~\singleclient \in \clientset \label{LPFLP_const1}\\ 
& \capacityU y_i \geq \sum_{j \in \clientset}{} x_{ij} \geq \B y_i & \forall~ \singlefacility \in \facilityset \label{LPFLP_const3}\\
& x_{ij} \leq y_i & \forall~ \singlefacility \in \facilityset , ~\singleclient \in \clientset \label{LPFLP_const4}\\   
& y_i,x_{ij} \in \left\lbrace 0,1 \right\rbrace  \label{LPFLP_const5}
\end{eqnarray}
where  $y_i$ is an indicator variable which is equal to $1$ if facility $i$ is open and is $0$ otherwise. $x_{ij}$ is an indicator variable which is equal to $1$ if client $j$ is served by facility $i$ and is $0$ otherwise. Constraints \ref{LPFLP_const1} ensure that every client is served. 
 Constraints \ref{LPFLP_const3} make sure that the total demand assigned to an open facility is at least $\B$ and at most $\capacityU$. Constraints \ref{LPFLP_const4} ensure that a client is assigned to a facility only if it is opened. LP-Relaxation  is obtained by allowing the variables to be non-integral.
Let $\zeta^{*} = <x^*, y^*>$ be an optimal solution to the LP-relaxation and $\opt{}$ be its cost.


We start by sparsifying the problem instance by removing some clients.
For $j \in \clientset$, let $ \C{j}= \sum_{i \in \facilityset} x^*_{ij} \dist{i}{j}$ denote the average connection cost paid by $\singleclient$ in $\zeta^*$. 
Further, let $\ell \ge 2$ be a tuneable parameter, $\ballofj{j}$ be the ball of facilities within a radius of $\ellone\C{j}$ of $j$ and $\cardTwo{\ballofj{j}}$ be the total extent up to which facilities are opened in $\ballofj{j}$ under solution $\zeta^*$ , i.e.,  $\cardTwo{\ballofj{j}} = \sum_{i \in \ballofj{j}} y^*_i $.
Then, $\cardTwo{\ballofj{j}} \ge (1-\frac{1}{\ellone}) \ge 1/2$.
The clients are processed in the non-decreasing order of the radii of their balls, removing the close-by clients with balls of larger radii and dissolving their balls: let $\bar{\clientset} = \clientset$ and $\clientset'$ denote the sparsified set of  clients. Initially $\clientset' = \phi$.
Let $j'$ be a client in $\bar{\clientset}$ with a ball of the smallest radius (breaking the ties arbitrarily). Remove $j'$ from $\bar{\clientset}$ and add it to $\clientset'$. For all $j (\ne j') \in \bar{\clientset}$ with $\dist{j'}{ j} \leq 2\ell \C{j}$, remove $j$ from $\bar{\clientset}$. Repeat the process until $\bar{\clientset} = \phi$. 
Cluster of facilities are formed around the clients in $\clientset'$ by assigning a facility to the cluster of $j' \in \clientset'$ if and only if $j'$ is nearest to the facility amongst all $j' \in \clientset'$, i.e. if $\bundle{j'}$ denotes the cluster centered at $j'$ then, $i$ belongs to $\bundle{j'}$ if and only if $\dist{i}{j'} < \dist{i}{k'}$ for all $k'(\ne j') \in \clientset'$ (assuming that the distances are distinct).
The clients in $C'$ are then called the cluster centers.

\begin{observation}
Any two cluster centers $j', k'$ in $C'$ satisfy the  separation property: $\dist{j'}{k'} > 2\ellone~max\{ \C{j'},\C{k'}\}$.
\end{observation}

\begin{lemma} \label{distBound}
Let $j' \in \clientset',~i \in \bundle{j'},~j \in \clientset$. Then,
\begin{enumerate}
    \item \label{claim-1} $\dist{i}{j'}  \leq \dist{i}{j} + 2\ellone \C{j}$
    \item \label{claim-2} $\dist{j}{j'} \leq 2\dist{i}{j} + 2\ellone \C{j}$
    \item \label{claim-3} If $\dist{j}{j'} \leq \ell \C{j'}$, then $\C{j'} \le 2\C{j}$.
\end{enumerate}
\end{lemma}

\begin{proof}
Let $j' \in \clientset',~i \in \bundle{j'},~j \in \clientset$. 
\begin{enumerate}

    \item Note that, $\dist{j}{k'} \leq 2\ellone \C{j}$ for some $k' \in \clientset'$. Then we have $\dist{i}{j'} \leq \dist{i}{k'} \leq \dist{i}{j} + \dist{j}{k'} \leq \dist{i}{j} + 2\ellone \C{j}$, where the first inequality follows because $i$ belongs to $\bundle{j'}$ and not $\bundle{k'}$ whenever $k' \ne j'$.  See Fig: (\ref{Cjj}).

    \item Using triangle inequality, we have $\dist{j}{j'} \leq \dist{i}{j} + \dist{i}{j'} \leq 2\dist{i}{j} + 2\ellone \C{j}$.
     \item Let $j \ne j'$. Note that  $\dist{j}{j'} \leq \ell \C{j'} \Rightarrow j \notin \clientset'$. 
     Suppose if possible, $\C{j'} > 2\C{j}$. Since $j \notin \cliset', \exists$ some $k'\in \clientset':\dist{j}{k'} \le 2 \ell \C{j}$. Then, $\dist{k'}{j'} \le \dist{k'}{j} + \dist{j}{j'}$ $ \leq 2 \ell \C{j} + \ell \C{j'} < 2\ell \C{j'}$. Thus, we arrive at a contradiction to the separation property. Hence, $\C{j'} \le 2\C{j}$.
\end{enumerate} \end{proof}

\begin{figure} 
    \centering
    \includegraphics[width=10cm]{ 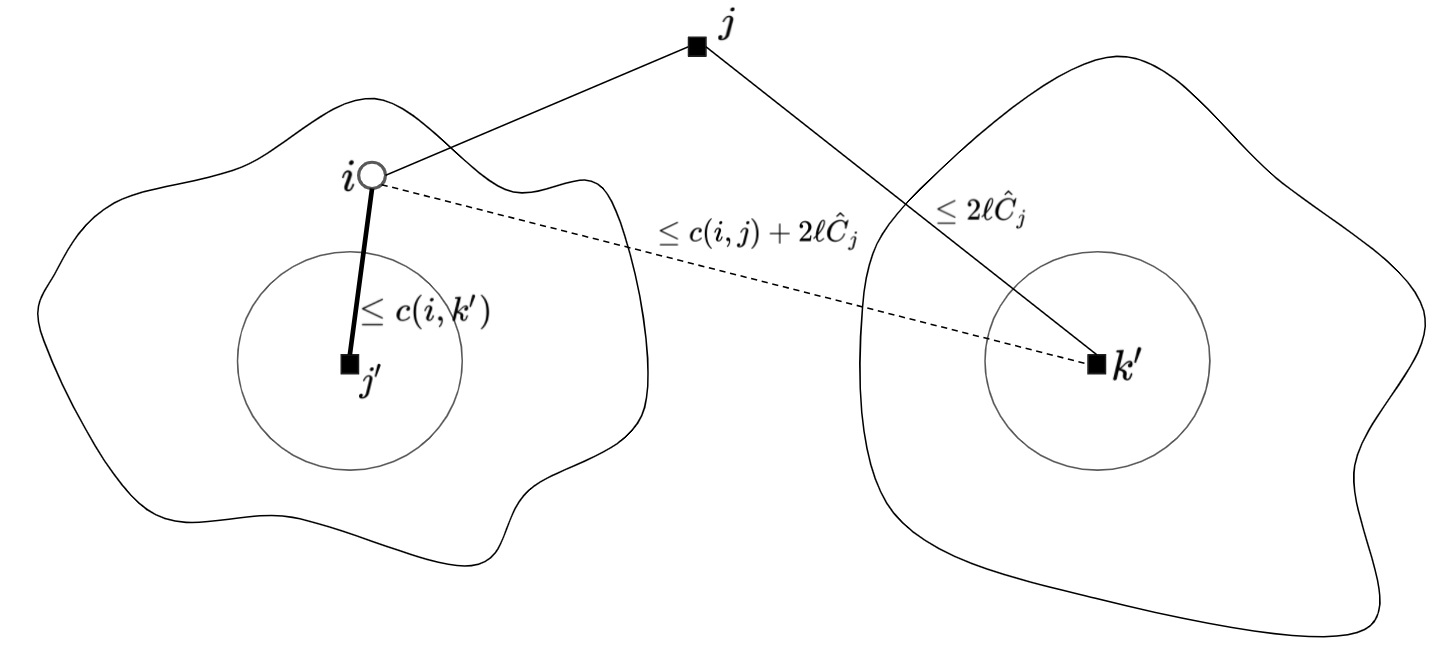}
    \caption{$\dist{i}{j'}  \leq \dist{i}{j} + 2\ellone \C{j}$}
    \label{Cjj}
\end{figure}


For $j' \in \clientset', j \in \clientset$, let $\loadjinc{j}{j'}$ be the extent up to which $j$ is served by the facilities in the cluster of $j'$ under solution $\zeta^*$ and $d_{j'}= \sum_{j \in \clientset} \loadjinc{j}{j'}$.
We call a cluster to be $sparse$ if $d_{j'} \le  \capacityU$ and $dense$ otherwise. Let $\csparse$ and $\cdense$ be the set of cluster centers of sparse  and dense clusters respectively.


\begin{lemma} \label{Sparse}
We can obtain a solution $<\hat{x}, \hat{y}>$ such that  exactly one facility $i(j')$ is opened integrally in each sparse cluster centered at $j' \in \clientset'$. The solution violates the lower bound by a factor of $(1 - 1/\ellone)$ and $\sum_{j' \in \clientset'}\sum_{i \in \bundle{j'}} [f_i \hat{y}_i + \sum_{j \in \clientset}  \hat{x}_{ij} \dist{i}{j}] \leq $
$\sum_{j' \in \csparse}[\frac{\ellone}{\ellone - 1}\sum_{i \in \bundle{j'}} f_i y^*_i + 4 \sum_{i \in \bundle{j'}} \sum_{j \in \clientset}  x^*_{ij} (\dist{i}{j} + \ellone \C{j})] + \sum_{j' \in \cdense} [\sum_{i \in \bundle{j'}} f_i y^*_i + \sum_{i \in \bundle{j'}} \sum_{j \in \clientset}  x^*_{ij} \dist{i}{j}]$.


\end{lemma}

\begin{proof}
For $ j' \in \cdense, i \in \bundle{j'}, j \in \clientset$, set $\hat{y}_i = y^*_i,  \hat{x}_{ij} = x^*_{ij}$. Next,
let $j' \in \csparse$
and $i(j')$ be a cheapest (lowest facility opening cost) facility in $\ballofj{j'}$. We open $i(j')$ and transfer all the assignments in the cluster onto it (see Figure (\ref{sparse})),
i.e. set $\hat{y}_{i(j')}=1, \hat{x}_{i(j')j} = \sum_{i \in \bundle{j'}}x^*_{ij}$, and $ \hat{y}_i =0$, $\hat{x}_{ij} =0$ for $i \in \bundle{j'} \setminus \{i(j')\}$ and $j \in \clientset$. 
Since $\cardTwo{\ballofj{j'}} \ge (1-\frac{1}{\ellone})$ we have $\sum_{j \in \clientset}{} \hat{x}_{i(j')j} = \demandofj{j'} \ge (1-\frac{1}{\ellone})\B$.
Thus, the lower bound is violated at most by $(1-\frac{1}{\ellone})$ 
and the facility cost is bounded by $\frac{\ellone}{\ellone - 1} \sum_{i \in \bundle{j'}} f_i y^*_i$.

To bound the service cost, we will show that for $j' \in \csparse, \sum_{j \in \clientset} \loadjinc{j}{j'} \dist{i(j')}{j}  \leq 4 \sum_{i \in \bundle{j'}} \sum_{j \in \clientset}  x^*_{ij} (\dist{i}{j} + \ellone \C{j})$: 
since $i(j') \in \ballofj{j'}$, we have $\dist{i(j')}{j'} \le \ellone \C{j'}$. 
Thus, for $j \in \clientset$,
we have $\dist{i(j')}{j} \le \dist{j'}{j} + \dist{i(j')}{j'} \le  \dist{j'}{j} + \ellone \C{j'}$. If $\ellone \C{j'} \le \dist{j'}{j}$ then $\dist{i(j')}{j}  \le 2  \dist{j'}{j}  \leq 4 (\dist{i}{j} + \ellone \C{j})$ ($\forall~i \in \bundle{j'}$ by claim (\ref{claim-2}) of Lemma (\ref{distBound})), 
 else $\dist{i(j')}{j}  \leq 2 \ell \C{j'} \leq 4 \ellone \C{j}$, where the second inequality in the else part follows by claim (\ref{claim-3}) of Lemma (\ref{distBound}).
Thus, in either case $\dist{i(j')}{j} \leq 4 (\dist{i}{j} + \ellone \C{j})$ for all $i \in \bundle{j'}$.
Substituting $\phi(j, j') = \sum_{i \in \bundle{j'}} x^*_{ij}$ and summing over all $j \in \clientset$ 
we get the desired bound.
Thus, $\sum_{i \in \facilityset} (f_i \hat{y}_i + \sum_{j \in \clientset} \hat{x}_{ij} \dist{i}{j} ) \\
= \sum_{j' \in \csparse} \sum_{i \in \bundle{j'}} f_i \hat{y}_i + \sum_{j' \in \csparse}  \sum_{i \in \bundle{j'}} \sum_{j \in \clientset}  \hat{x}_{ij}  \dist{i(j')}{j}  + \sum_{j' \in \cdense} \sum_{i \in \bundle{j'}} (f_i \hat{y}_i + \sum_{j \in \clientset} \hat{x}_{ij} \dist{i}{j} ) \\
= \sum_{j' \in \csparse} \sum_{i \in \bundle{j'}} f_i \hat{y}_i + \sum_{j' \in \csparse} \sum_{j \in \clientset}  \loadjinc{j}{j'} \dist{i(j')}{j}  + \sum_{j' \in \cdense} \sum_{i \in \bundle{j'}} (f_i \hat{y}_i + \sum_{j \in \clientset} \hat{x}_{ij} \dist{i}{j} ) \\
\le \sum_{j' \in \csparse}[\frac{\ellone}{\ellone - 1}\sum_{i \in \bundle{j'}} f_i y^*_i + 4 \sum_{i \in \bundle{j'}} \sum_{j \in \clientset}  x^*_{ij} (\dist{i}{j} + \ellone \C{j})] + \sum_{j' \in \cdense} \sum_{i \in \bundle{j'}} [ f_i y^*_i + \sum_{j \in \clientset}  x^*_{ij} \dist{i}{j}]$
\end{proof}
 
   
    \begin{figure} 
    \centering
    \includegraphics[width=10cm]{ 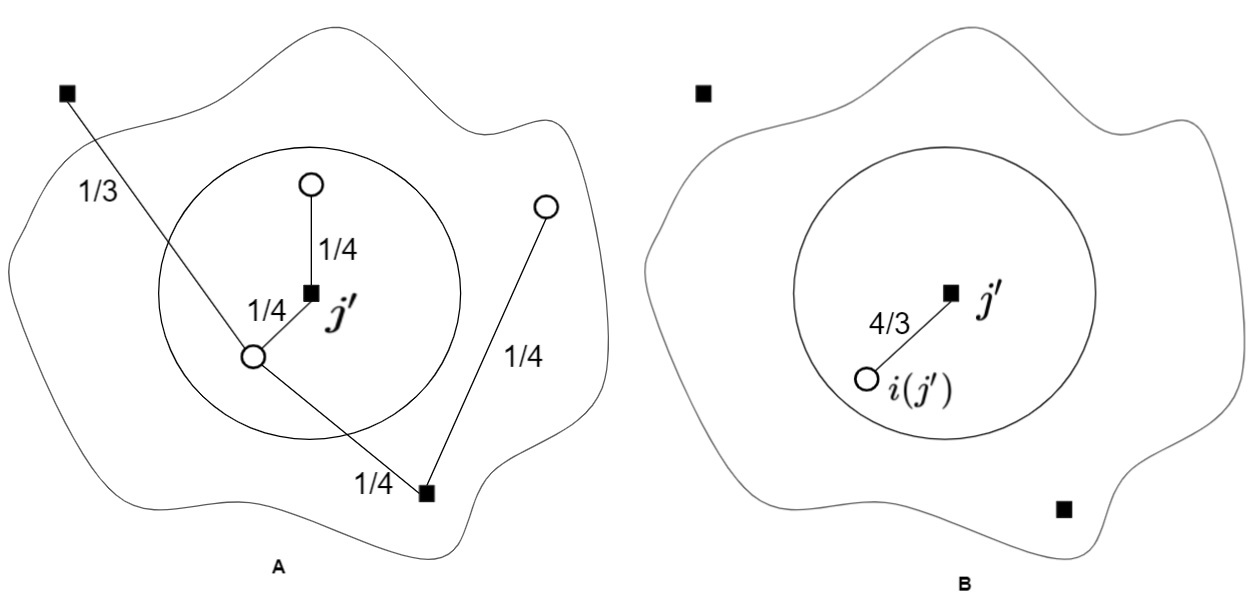}
    \caption{(A) Cluster $\bundle{j'}$ centered at $j'$. (B) Demand $\demandofj{j'}$ being assigned to facility $i(j') \in \ballofj{j'}$.}
    \label{sparse}
\end{figure}
 
To open the facilities integrally in the dense clusters, we consider the following LP for every dense cluster $\bundle{j'}$.

($LP_1$): 
\label{dense}
$Minimize ~\sum_{\singlefacility \in \bundle{j'}} (f_i +\capacityU \dist{i}{j'}) \zofi{i}$
\begin{eqnarray}
subject~ to &\capacityU \sum_{\singlefacility \in \bundle{j'}} \zofi{i} \geq \demandofj{j'} \label{denseLP_const1}\\
&\B \sum_{\singlefacility \in \bundle{j'}} \zofi{i} \leq \demandofj{j'} \label{denseLP_const1-1}\\ 
&0 \leq \zofi{i} \leq 1 \label{denseLP_const2}
\end{eqnarray}
   

Note that for $\zofi{i} = \sum_{j \in \clientset}x_{ij}^*/\capacityU$, we have $\demandofj{j'} = \capacityU \sum_{\singlefacility \in   \bundle{j'}} \zofi{i} \geq \B \sum_{\singlefacility \in \bundle{j'}} \zofi{i}$ and $\zofi{i} \le y^*_i$. Thus, $z$
is a feasible solution with cost at most $ \sum_{i \in \bundle{j'}} [\facilitycost y^*_i + \sum_{ j \in \clientset} x^*_{ij} ( \dist{i}{j} +  2\ellone \C{j}) ]$ by claim (\ref{claim-1}) of Lemma (\ref{distBound}).


\begin{lemma} \label{integral_dense}
Given the feasible solution $z$ to $LP_1$, an integral solution $\hat{z}$, that violates constraint (\ref{denseLP_const1}) and (\ref{denseLP_const1-1}) by a factor of 
$(2-1/\ell)$ 
and $(1-1/\ell)$ 
respectively, can be obtained at a loss of factor $\ell/(\ell-1)$ in cost
i.e. $ \sum_{\singlefacility \in \bundle{j'}} (f_i +\capacityU \dist{i}{j'}) \hat{z}_{i} \le \ell/(\ell-1) \sum_{\singlefacility \in \bundle{j'}} (f_i +\capacityU \dist{i}{j'}) \zofi{i}$.
\end{lemma}

\begin{proof}

We say that a solution to $LP_1$ is almost integral if it has at most one fractionally opened facility in $\bundle{j'}$.
We first obtain an almost integral solution
$\primezofi{}$ by arranging the facilities opened in $z$, in non-decreasing order of $\facilitycost  + \dist{i}{j'} \capacityU$ and greedily transferring the openings $z$ onto them. 
Then,  $\B \sum_{i \in \bundle{j'}} \primezofi{i} \leq \capacityU \sum_{i \in \bundle{j'}} \primezofi{i} = \capacityU \sum_{i \in \bundle{j'}} \zofi{i} = \demandofj{j'}$. 
Note that the cost of solution $\primezofi{}$ is no more than that of solution $z$.

We now convert the almost integral solution $\primezofi{}$ to an integral solution $\hat{z}$. Let $\hat{z} = \primezofi{}$ initially. Let $i'$ be the fractionally opened facility, if any, in $\bundle{j'}$. 
Consider the following cases:

\begin{enumerate}
    \item $\primezofi{i'} \leq 1-1/\ell$ 
    : close $i'$ in $\hat{z}$. There must be at least one integrally opened facility, say $i (\ne i') \in \bundle{j'}$ in $\primezofi{}$, as $\demandofj{j'} \ge \capacityU$. Then, $\demandofj{j'} = \capacityU \sum_{k \in \bundle{j'} \setminus \{i, i'\}} \primezofi{k} + \capacityU  (\primezofi{i} + \primezofi{i'}) \le \capacityU \sum_{k \in \bundle{j'} \setminus \{i, i'\}} \primezofi{k}  + 
    (2-1/\ell) \capacityU  \primezofi{i} \le 
    (2-1/\ell) \capacityU \sum_{k \in \bundle{j'} \setminus \{i'\}} \primezofi{k} = 
    (2-1/\ell) \capacityU \sum_{k \in \bundle{j'}} \hat{z}_k$. There is no increase in cost as we have only (possibly) shut down one of the facilities. Also, $\demandofj{j'} \geq \B \sum_{i \in \bundle{j'}} \primezofi{i} > \B \sum_{i \in \bundle{j'}} \hat{z}_{i}$ as $\sum_{i \in \bundle{j'}} \primezofi{i} > \sum_{i \in \bundle{j'}} \hat{z}_{i}$.

    \item $\primezofi{i'} 
    > 1-1/\ell$ 
    : Open $i'$ integrally at a loss of factor 
    $(\ell/\ell-1)$ in facility opening cost i.e., $\sum_{i \in \bundle{j'}} f_i \hat{z}_i \le 
    (\ell/\ell-1) \sum_{i \in \bundle{j'}} f_i \primezofi{i}$ and $(1-1/\ell)$ factor in lower bound, i.e., $\demandofj{j'} \geq \B \sum_{k \in \bundle{j'}} \primezofi{k} \geq \B \frac{\ell-1}{\ell} \sum_{k \in \bundle{j'}} \hat{z}_{k}$,  where the second inequality follows because $\sum_{i \in \bundle{j'}} \hat{z}_i \leq \frac{\ell}{\ell-1}\sum_{i \in \bundle{j'}} z'_i$. Also,  $\demandofj{j'} = \capacityU \sum_{i \in \bundle{j'}} \primezofi{i} < \capacityU \sum_{i \in \bundle{j'}} \hat{z}_{i}$ as $\sum_{i \in \bundle{j'}} \primezofi{i} < \sum_{i \in \bundle{j'}} \hat{z}_{i}$.
\end{enumerate}\end{proof}



Next, we define our assignments, possibly fractional, in the dense clusters. For $j' \in \cdense$, we distribute $d_{j'}$ equally to the facilities opened in $ \hat{z}$. Let $l_i$ be the amount assigned to  facility $i$ under this distribution. Then, $l_i = \hat{z}_i \frac{\demandofj{j'}}{\sum_{i \in \bundle{j'}} \hat{z}_i} \leq \hat{z}_i \frac{
(2-1/\ell)\capacityU \sum_{i \in \bundle{j'}} \hat{z}_i}{\sum_{i \in \bundle{j'}} \hat{z}_i } = 
(2-1/\ell) \capacityU \hat{z}_i$, where the first inequality follows by Lemma (\ref{integral_dense}). Also, $l_i  = \hat{z}_i \frac{\demandofj{j'}}{\sum_{i \in \bundle{j'}} \hat{z}_i} \geq \hat{z}_i \frac{\demandofj{j'}}{\frac{\ell}{\ell-1}\sum_{i \in \bundle{j'}} z'_i} = \hat{z}_i \frac{\demandofj{j'}}{\frac{\ell}{\ell-1}\sum_{i \in \bundle{j'}} z_i} \geq  \hat{z}_i \frac{\B \sum_{i \in \bundle{j'}} z_i}{\frac{\ell}{\ell-1}\sum_{i \in \bundle{j'}} z_i} = \frac{\ell-1}{\ell} \B \hat{z}_i$, where the first inequality follows because $\sum_{i \in \bundle{j'}} \hat{z}_i \leq \frac{\ell}{\ell-1}\sum_{i \in \bundle{j'}} z'_i$.

A solution is said to be an integrally opened solution if all facilities in it are opened to an extent of either $1$ or $0$. We next obtain such a solution to LBUBFL problem.

\begin{lemma} \label{Dense1}
We can obtain an integrally opened solution $<\bar x, \bar y>$ to LBUBFL problem at a loss of 
$(2-1/\ell)$ in upper bounds and $(1-1/\ell)$ in lower bounds whose cost is bounded by 
$(10\ellone + 4)LP_{opt}$.
\end{lemma}

\begin{proof}
For $j' \in \csparse$, set $\bar y_i = \hat{y}_i,~\bar{x}_{ij} = \hat{x}_{ij}~\forall i \in \bundle{j'}, j \in \clientset$. The violation in lower bound and the following cost bound then follows from Lemma (\ref{Sparse}).

\begin{equation}
\label{cost-of-sparse}
\sum_{j' \in \csparse} \sum_{i \in \bundle{j'}} (f_i \bar{y}_i + \sum_{j \in \clientset} \bar{x}_{ij} \dist{i}{j}) \le \sum_{j' \in \csparse}[\frac{\ellone}{\ellone - 1}\sum_{i \in \bundle{j'}} f_i y^*_i + 4 \sum_{i \in \bundle{j'}} \sum_{j \in \clientset}  x^*_{ij} (\dist{i}{j} + \ellone \C{j})].
\end{equation}

Next let $j' \in \cdense,~i \in \bundle{j'}, j \in \clientset$. Set $\bar y_i = \hat{z}_i,~\bar{x}_{ij} = \frac{l_i}{\demandofj{j'}} \loadjinc{j}{j'}$. 
Consider $ \sum_{j \in \clientset} \bar{x}_{ij} = \sum_{j \in \clientset} \frac{l_i}{\demandofj{j'}} \loadjinc{j}{j'} = l_i \leq   
(2-1/\ell)\capacityU \hat{z}_i = 
(2-1/\ell)\capacityU \bar{y}_i$. Also, $\sum_{j \in \clientset} \bar{x}_{ij} = l_i \geq   
(1-1/\ell)\B \hat{z}_i = (1-1/\ell)\B \bar{y}_i $. Thus, the loss in upper bound and lower bound is at most 
$(2-1/\ell)$ and $(1-1/\ell)$ respectively.


Also,\\ 
$\sum_{\singlefacility \in \bundle{j'}} (f_i \bar{y}_i +  \sum_{j \in \clientset} \dist{i}{j} \bar{x}_{ij}) $\\
$= \sum_{\singlefacility \in \bundle{j'}} f_i \hat{z}_i +  \sum_{\singlefacility \in \bundle{j'}}  \sum_{j \in \clientset} \dist{i}{j}  \frac{l_i}{\demandofj{j'}} \loadjinc{j}{j'}$\\
$ \leq \sum_{\singlefacility \in \bundle{j'}} f_i \hat{z}_i +  \sum_{\singlefacility \in \bundle{j'}}  \sum_{j \in \clientset} \frac{l_i}{\demandofj{j'}} \loadjinc{j}{j'} (\dist{j}{j'} + \dist{i}{j'})$ (by triangle inequality)\\ 
$ = \sum_{\singlefacility \in \bundle{j'}} f_i \hat{z}_i +  \sum_{\singlefacility \in \bundle{j'}}  l_i \dist{i}{j'} +   
\sum_{j \in \clientset}  \loadjinc{j}{j'} \dist{j}{j'} $ (as $ \sum_{\singlefacility \in \bundle{j'}} l_i = \sum_{j \in \clientset} \loadjinc{j}{j'}  = \demandofj{j'}$)\\
$ \leq  
\ell/(\ell-1) \sum_{\singlefacility \in \bundle{j'}} f_i \primezofi{i} +    
(2-1/\ell) \capacityU \sum_{\singlefacility \in \bundle{j'}} \dist{i}{j'} \primezofi{i} +  
\sum_{j \in \clientset}  \loadjinc{j}{j'} \dist{j}{j'} $ \\ 
$\leq  max\{\ell/(\ell-1), (2-1/\ell)\} \sum_{\singlefacility \in \bundle{j'}} (f_i  +  \capacityU \dist{i}{j'} )\primezofi{i} + 
\sum_{j \in \clientset}  \loadjinc{j}{j'} \dist{j}{j'}$ \\
$\leq  
2  \sum_{\singlefacility \in \bundle{j'}} (f_i  +  \capacityU \dist{i}{j'} )\primezofi{i} + \sum_{j \in \clientset}  \loadjinc{j}{j'} \dist{j}{j'}$ \\ 

$\leq 
2  \sum_{\singlefacility \in \bundle{j'}} (f_i  +  \capacityU \dist{i}{j'} ) z_i + \sum_{j \in \clientset}  \loadjinc{j}{j'} \dist{j}{j'}$ (by Lemma (\ref{integral_dense}))\\
$\leq  
2 \sum_{\singlefacility \in \bundle{j'}} f_i y^*_i  +   
2  \sum_{\singlefacility \in \bundle{j'}} \sum_{j \in \clientset} \dist{i}{j'} x^*_{ij} + \sum_{j \in \clientset}  \loadjinc{j}{j'} \dist{j}{j'}$ (because $z_i = \sum_{j \in \clientset} x^*_{ij} / \capacityU \leq y^*_i$)\\ 
$\leq  
2  \sum_{\singlefacility \in \bundle{j'}} f_i y^*_i  +  
2 \sum_{\singlefacility \in \bundle{j'}} \sum_{j \in \clientset}  x^*_{ij} (\dist{i}{ j} + 2\ell \C{j}) + \sum_{j \in \clientset}  \sum_{i \in \bundle{j'}} x^*_{ij}  (2\dist{i}{j} + 2\ell\C{j} )$ 
(by definition of $\loadjinc{j}{j'}$ and claims (\ref{claim-1}) and (\ref{claim-2}) of Lemma (\ref{distBound}))\\
$= 
2  \sum_{\singlefacility \in \bundle{j'}} f_i y^*_i  +   
4   \sum_{\singlefacility \in \bundle{j'}} \sum_{j \in \clientset}  x^*_{ij} \dist{i}{ j} +   
6\ell \sum_{j \in \clientset}  \sum_{i \in \bundle{j'}} x^*_{ij}   \C{j} $ \\

Summing over all $j' \in \cdense$ and adding inequality (\ref{cost-of-sparse}), we get\\
$ \sum_{i \in \facilityset} (f_i \bar{y}_i +  \sum_{j \in \clientset} \dist{i}{j} \bar{x}_{ij})$\\
$\leq max\{2, \ell/(\ell-1)\} \sum_{i \in \facilityset}  f_i y^*_i + 
4 \sum_{i \in \facilityset}\sum_{j \in \clientset}  x^*_{ij} \dist{i}{j} + 
10 \ell \sum_{i \in \facilityset}\sum_{j \in \clientset} x^*_{ij} \C{j}$\\
  
$ \leq 2 \sum_{i \in \facilityset}  f_i y^*_i + 4\sum_{i \in \facilityset}  \sum_{j \in \clientset}  x^*_{ij} \dist{i}{j} + 10\ell \sum_{i \in \facilityset}  \sum_{j \in \clientset}  x^*_{ij} \C{j}$ (for $\ell \geq 2$)\\
$ \leq 
(10 \ell+4)LP_{opt}$.
\end{proof}

Theorem~\ref{TricriteriaResult} is obtained by solving a min-cost flow problem with relaxed lower  and upper bounds to obtain the integral assignments.


For $\ell = 2.01$, we get $\alpha > 1/2$, $\beta$  slightly more than $3/2$ and the approximation ratio less than $25$. $\beta$ can be reduced to $3/2$ by a slight modification in obtaining an integral solution $\hat z$ from an almost integral solution $\primezofi{}$ in Lemma~(\ref{integral_dense}): instead of comparing $\primezofi{i'}$ with $(1 - 1/\ell)$, we compare it with $1/2$.

Let $S^t$ be the solution so obtained with $\alpha > 1/2$, $\beta = 3/2$ . 
Then, $Cost_I(S^t) \le 
O(1)Cost_I(O)$,
where $O$ is an optimal solution to $I$. 
%
	Next, using $S^t$, we transform the instance $I$ to an instance $I_{cap}$ of capacitated facility location problem by swapping the roles of clients and facilities. This is done via a series of transformations from instance $I$ to $I_1$, 	$I_1$ to $I_2$ and $I_2$ to $I_{cap}$. The key idea is to create  $\B - n_i$ units of demand, where the number $n_i$ of clients served by facility $i$ in $S^t$ is short of the lower bound and create  $n_i - \B$ units of supply at locations where $n_i > \B$.

%% file: LBUBFLGeneralTree.tex
\section{Instance $I_1$ and $I_2$}
\label{instI1}

In this section, we first transform instance $I$ to instance $I_1$ ($\facilityset,~\clientset,~f^1,~c,~\B,~\capacityU$) of LBUBFL by moving the clients to the facilities serving them in the tri-criteria solution $S^t$ and then transform $I_1$ to an instance $I_2$ of LBFL by removing the facilities not opened in $S^t$.
Recall that for a client $j$,  $\sigma^t(j)$ is the facility  in $\facilityset^t$ serving $j$. For $i \in \facilityset^t$, let $n_i$ be the number of clients served by $i$ in $S^t$, i.e., $n_i = |(\sigma^t)^{-1}(i)|$ and for $i \notin \facilityset^t$, $n_i = 0$. Move these clients to $i$ (see Fig.~\ref{Mapping1}). Thus, there are $n_i$ clients co-located with $i$. In $I_1$, our facility set is $\facilityset$ and the clients are at their new locations. Facility opening costs in $I_1$ are modified as follows: $f^1_i = 0$ for $i \in \facilityset^t$ and $=f_i$ for $i \notin \facilityset^t$.

\begin{lemma}
Cost of optimal solution of $I_1$ is bounded by $Cost_I(S^t) + Cost_I(O)$.
\end{lemma}
\begin{proof}
We construct a feasible solution $S_1$ of $I_1$: open $i$ in $S_1$ iff it is opened in $O$. Assign $j$ to $i$ in $S_1$ iff it is assigned to $i$ in $O$. $S_1$ satisfies both the lower bound as well as the upper bound as $O$ does so. For a client $j$, let $\sigma^*(j)$ be the facility serving $j$ in $O$. Then, the cost $\dist{ \sigma^t(j)}{\sigma^*(j))}$ of serving $j$ from its new location, $\sigma^t(j)$ is bounded by $\dist{j}{\sigma^t(j)} +\dist{j}{\sigma^*(j)}$ (see Fig. \ref{lemma1}). Summing over all $j \in \clientset$ and adding the facility opening costs, we get the desired claim.

\begin{figure}[t]
    \centering
    \includegraphics[width=9cm]{ 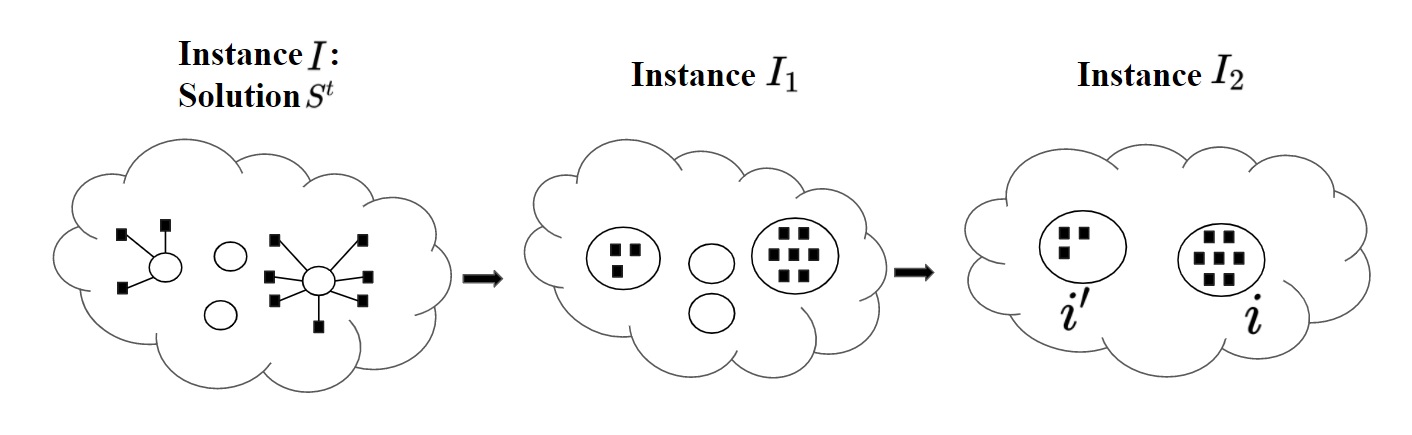}
    \caption{Transforming instance $I$ to $I_1$ and  $I_1$ to $I_2$. Solid squares represent clients and circles represent facilities. Clients are moved to the facilities serving them in $S^t$ while transforming $I$ to $I_1$. Facilities not opened in $S^t$ are dropped while going from $I_1$ to $I_2$.}
    \label{Mapping1}
\end{figure}

\begin{figure}[t]
    \centering
    \includegraphics[width=4cm]{ 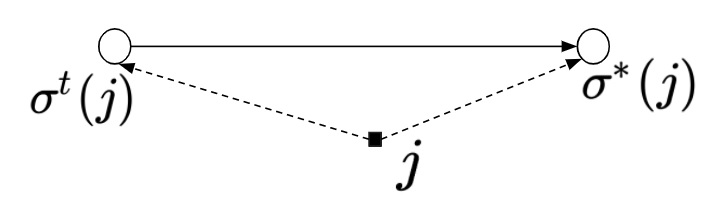}
    \caption{$c( \sigma^t(j), \sigma^*(j)) \leq c(j, \sigma^t(j)) + c(j, \sigma^*(j))$ }
    \label{lemma1}
\end{figure}
\end{proof}

    
   Next, we define an instance ~$I_2(\facilityset^t,~\clientset,~f^2,~c,~\B)$ of LBFL from $I_1$ by removing the facilities not in $\facilityset^t$ and ignoring the upper bounds. Instance $I_2$ is same as $I_1$ without the upper bounds, except that the facility set is $\facilityset^t$ now. Thus, $f^2_i = 0~\forall~i \in \facilityset^t$. 
  Let $O_1$ be an optimal solution of $I_1$. 
  For a client $j$, let $\sigma^1(j)$ be the facility serving $j$ in $O_1$.

\begin{lemma}
Cost of optimal solution of $I_2$ is bounded by $ 2Cost_{I_1}(O_1)$.
\end{lemma}
\begin{proof}
We construct a feasible solution $S_2$ to $I_2$: for $i \in \facilityset^t$, open $i$ in $S_2$ if it is opened in $O_1$ and assign a client $j$ to it if it is assigned to $i$ in $O_1$. For a facility $i \notin \facilityset^t$, 
open its nearest facility $i'$ (if not already opened) in $\facilityset^t$. Assign $j$ to $i'$ in $S_2$ if it is assigned to $i$ in $O_1$. $S_2$ satisfies the lower bound  as $O_1$ does so. Note that it need not satisfy the upper bounds. This is the reason we dropped upper bounds in $I_2$. The cost $c( \sigma^t(j), i')$
of serving $j$ from its new location is bounded by $c(\sigma^t(j), \sigma^1(j)) + c( \sigma^1(j), i') \le 2 c(\sigma^t(j), \sigma^1(j))$ (since $i'$ is closest in $\facilityset^t$ to $i =  \sigma^1(j)$). Summing over all $j \in \clientset$ we get the desired claim.
\end{proof}


\section{Instance $I_{cap}$ of Capacitated Facility Location Problem}
\label{instIcap}

 In this section, we create an  instance $I_{cap}$ of capacitated facility location problem from $I_2$. 
 The main idea is to create  a demand of $\B - n_i$ units at locations where the number of clients served by the facility is less than $\B$ and a supply of $n_i - \B$ units at locations with surplus clients. For each facility $i \in \facilityset^t$, let $l(i)$ be the distance of $i$ from a nearest facility  $i' \in \facilityset^t$, $i' \neq i$ and let $\delta$ be a constant to be chosen appropriately. A facility $i$ is called {\em small} if $0< n_i \leq \B$ and  {\em big} otherwise. A big facility $i$ is split into two co-located facilities $i_1$ and $i_2$.
  We also split the set of clients at $i$ into two sets:  arbitrarily, $\B$ of these clients are placed at $i_1$ and the remaining $n_i - \B$ clients at $i_2$ (see Fig.~\ref{Mapping2}). 
   Instance $I_{cap}$ is then defined as follows:
A {\em small} facility $i$ needs additional $\B - n_{i}$ clients to satisfy its lower bound; hence a demand of $\B - n_{i}$ is created at $i$ in the $I_{cap}$ instance. A facility with capacity $\B$ and facility opening cost $\delta n_{i} l(i)$ is also created at $i$. 
For a {\em big} facility $i$, correspondingly two co-located facilities $i_1$ and $i_2$ are created with capacities $\B$ and $n_i - \B$ respectively. The facility opening cost of $i_1$ is $\delta \B l(i)$ whereas $i_2$ is free. Intuitively, since the lower bound of a big facility is satisfied, it has some extra ($n_i - \B$) clients, which can be used to satisfy the demand of clients in $I_{cap}$ for free. 
 The second type of big facilities are called {\em free}. 
 We use $i$ to refer to both the client (with demand) as well as the facility located at $i$. Let $\bar \facilityset^t$ be the set of facilities so obtained. The set of clients and the set of facilities are both $\bar \facilityset^t$. Table~\ref{tab:my_label} summarizes the instance. Also see Fig.~\ref{Mapping2}
 
 \begin{figure}
    \centering
    \includegraphics[width=8cm]{ 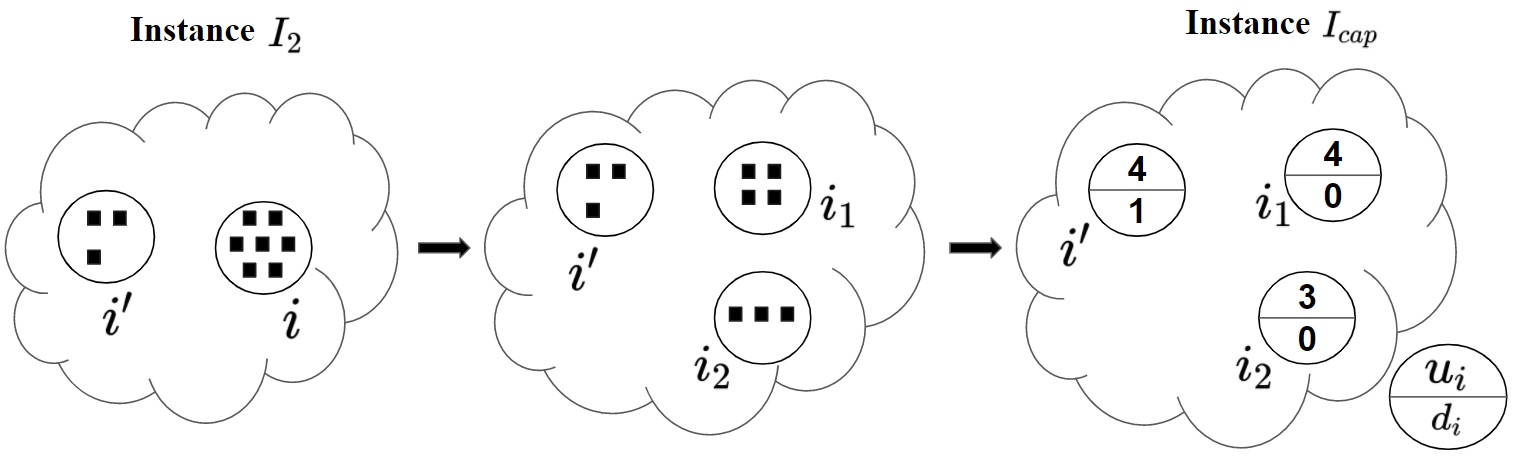}
    \caption{Transforming instance $I_2$ to $I_{cap}$: $\B = 4$,
    $i'$ is small, $i$ is big.
    $i$ is split into $i_1$ and $i_2$ while transforming $I_2$ to $I_{cap}$, $i_2$ is free. Demand at the small facility $i'$ is $\B - n_{i'} = 1$ and at $i_1$, $i_2$, it is $0$. Capacities of $i'$ and $i_1$ are $\B = 4$ and it is $n_{i} - \B = 3$ of the free facility $i_2$.}
    \label{Mapping2}
\end{figure}
 
 \begin{table}[t]
    \centering
    \begin{tabular}{|c|c|c|c|c|}
    \hline
       Type & $n_i$ & $u_i$ & $d_i$ & $f^t_i$ \\
        \hline
       {\em small} & $n_i \leq \B$ & $\B$ & $\B- n_i$& $\delta n_i l(i)$ \\
        \hline
       {\em big} & $n_i > \B$ & ($i_1$)$\B$ & $0$& $\delta \B l(i)$ \\
        & & ($i_2$) $n_i - \B$ & $0$& $0$ \\
       
          \hline 
    \end{tabular}
    \caption{Instance of $I_{cap}$: $d_i, u_i, f^t_i$ are demands, capacities and facility costs resp.}
    \label{tab:my_label}
\end{table}
 
 


\begin{lemma}
Let $O_2$ be an optimal solution of $I_2$. Then, cost of optimal solution to $I_{cap}$ is bounded by $(1+2\delta)Cost_{I_2}(O_2)$.
\end{lemma}

\begin{proof}
We will construct a feasible solution $S_{cap}$ to $I_{cap}$ of bounded cost from $O_2$.
 As $O_2$ satisfies the lower bound,  we can assume wlog that if $i$ is opened in $O_2$ then it serves all of its clients if $n_i \le \B$ (before taking more clients from outside) and it serves at least $\B$ of its clients before sending out its clients to other (small) facilities otherwise. 
Also, if two big facilities are opened in $O_2$, one does not serve the clients of the other. Let $\rho^2(j_i, i')$ denote the number of clients co-located at $i$ and assigned to $i'$ in $O_2$ and $\rho^c(j_{i'}, i)$ denotes the amount of demand of $i'$ assigned to $i$ in $S_{cap}$.
 \begin{figure}[t]
    \centering
    \includegraphics[width=10cm]{ 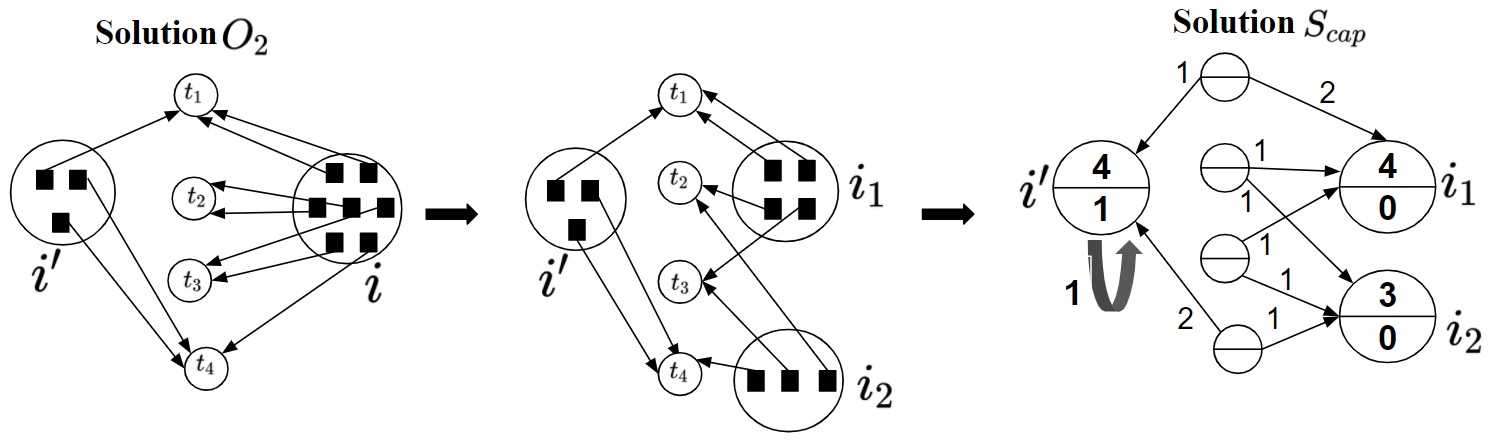}
    \caption{Feasible solution $S_{cap}$ to $I_{cap}$ from $O_2$: $i$ and $i'$ are closed in $O_2$. $\rho^2(j_{i_1}, t_1) = 2,~\rho^2(j_{i_1}, t_2) = 1,~\rho^2(j_{i_1}, t_3) = 1,~ \rho^2(j_{i_1}, t_4) = 0$ and $\rho^2(j_{i_2}, t_1) = 0,~\rho^2(j_{i_2}, t_2) = 1,~\rho^2(j_{i_2}, t_3) = 1,~ \rho^2(j_{i_2}, t_4) = 0$.  } 
    \label{Mapping}
\end{figure}

 \begin{figure}[t]
    \centering
    \includegraphics[width=8cm]{ 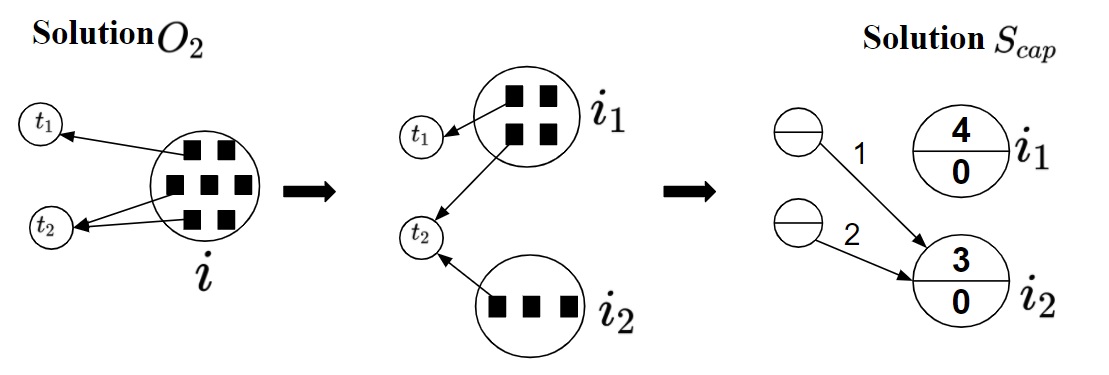}
   \caption{Feasible solution $S_{cap}$ to $I_{cap}$ from $O_2$: $i$ is open in $O_2$.   $\rho^c(j_{t_1}, i_2) = \rho^2(j_{i}, t_1) = 1, \rho^c(j_{t_2}, i_2) = \rho^2(j_{i}, t_2) = 2$,
   $i_2$ is open in $S_{cap}$. $t_1$ and $t_2$ must be small.}
    \label{O2_SCap2}
\end{figure}
    \begin{enumerate}
        \item   if $i$ is closed in $O_2$ then open $i$ if $i$ is small
        and open $(i_1 \& i_2)$ if it is big.
        
        Assignments are defined as follows (see Fig.~\ref{Mapping}):
        \begin{enumerate}
        
        \item If $i$ is small: Let $i$ serve its own demand. In addition, assign $\rho^2(j_i, i')$ demand of $i'$ to $i$ in our solution for small $i' \ne i$. Note that the same cannot be done when $i'$ is big as there is no demand at $i'_1$ and $i'_2$. Thus,  $\rho^c(j_{i'}, i) = \rho^2(j_i, i')$, for small $ i' \ne i$ and $\rho^c(j_i, i) = d_i$. Also, $\sum_{small~i' \ne i}\rho^c(j_{i'}, i) = \sum_{small~i' \ne i} \rho^2(j_i, i') \le n_i$.
        Hence, $\sum_{i'}\rho^c(j_{i'}, i) \le n_i + d_i = \B = u_i$. 
      
        
       \item If $i$ is big: assign $\rho^2(j_{i_1}, i')$ (/$\rho^2(j_{i_2}, i')$) demand of $i'$ to $i_1$(/$i_2$) in our solution for small $i' \ne i$. Thus,
        $\sum_{small~i' \ne i}\rho^c(j_{i'}, i_1) = \sum_{small~i' \ne i} \rho^2(j_{i_1}, i') \le \B = u_{i_1}$. Also,
         $\sum_{small~i' \ne i}\rho^c(j_{i'}, i_2) \\ = \sum_{small~i' \ne i} \rho^2(j_{i_2}, i') \le n_i - \B = u_{i_2}$.
     
       \end{enumerate}
     
     \item If $i$ is opened in $O_2$ and is big, open the free facility $i_2$ (see Fig.~\ref{O2_SCap2}): assign $\rho^2(j_{i}, i')$  demand of $i'$ to $i_2$ in our solution for small $i' \ne i$. Thus,
        $\sum_{i' \ne i}\rho^c(j_{i'}, i_2) = \sum_{i' \ne i} \rho^2(j_{i}, i') \le n_i - \B$ (the inequality holds by assumption on $O_2$) $= u_{i_2}$.
    \end{enumerate}
    
    Thus capacities are respected in each of the above cases. Next, we show that all the demands are satisfied. If a facility is opened in $S_{cap}$, it satisfies its own demand. 
Let $i$ be closed in $S_{cap}$. If $i$ is big ,we need not worry as $i_1$ and $i_2$ have no demand. So, let $i$ is small. Then, it must be opened in $O_2$. Then, 
$\sum_{i' \ne i}\rho^c(j_{i}, i') = $
$\sum_{small ~i' \ne i,  }\rho^c(j_{i}, i') + \sum_{i'_1: i' ~is ~big}\rho^c(j_{i}, i'_1) + 
\sum_{i'_2: i' ~is ~big}\rho^c(j_{i}, i'_2)= $
$\sum_{small ~i' \ne i,  }\rho^2(j_{i'}, i) + \sum_{i'_1 : i' ~is ~big} \rho^2(j_{i'_1}, i)  + \\ \sum_{i'_2 : i' ~is ~big}\rho^2(j_{i'_2}, i)= \sum_{small ~i' \ne i,  }\rho^2(j_{i'}, i) + \sum_{i' ~is ~big} \rho^2(j_{i'}, i) = \sum_{i' \ne i}\rho^2(j_{i'}, i)$
$\ge \B - n_i$  (since $O_2$ is a feasible solution of $I_2$) $= d_i$.

Next, we bound the cost of the solution. The connection cost is at most that of $O_2$. 
For facility costs, consider a facility $i$ that is opened in our solution and closed in $O_2$. Such a facility must have paid a cost of at least $n_i l(i)$ to get its clients served by other (opened) facilities in $O_2$. The facility cost paid by our solution is $ \delta min\{n_i, \B\} l(i) \le \delta n_i l(i)$. If $i$ is opened in $O_2$, then it must be serving at least $\B$ clients and hence paying a cost of at least $\B l(i)$ in $O_2$. In this case also, the facility cost paid by our solution is $ \delta min\{n_i, \B\} l(i) \le \delta \B l(i)$.  Summing over all $i$'s, we get 
that the facility cost is bounded by $2 \delta Cost_{I_2}(O_2)$ and the total cost is bounded by $(1 + 2 \delta) Cost_{I_2}(O_2)$. Factor $2$ comes because we may have counted an edge $(i, i')$ twice, once as a client when $i$ was closed and once as a facility when $i'$ was opened in $O_2$.  
\end{proof}


\section{Approximate Solution $AS_1$ to $I_1$ from approximate solution $AS_{cap}$ to  $I_{cap}$} \label{bi-criteria}

In this section,  we obtain
 a solution $AS_1$ to $I_1$ that violates the upper bounds by a factor of $(\beta + 1)$ without violating the lower bounds. This is the second major contribution of our work. We first obtain a ($5+\epsilon$)-approximate solution $AS_{cap}$ to $I_{cap}$ using approximation algorithm of Bansal \etal~\cite{Bansal}.  $AS_{cap}$ is then used to construct $AS_1$. Wlog assume that if a facility $i$ is opened in $AS_{cap}$ then it serves all its demand. (This is always feasible as $d_i \le u_i$.)  If this is not true, we can modify $AS_{cap}$  and obtain another solution, that satisfies the condition, of cost no more than that of $AS_{cap}$.
 $AS_1$ is obtained from $AS_{cap}$ by first defining the assignment of the clients and then opening the facilities that get at least $\B$ clients. Let $\rhobar^c(j_{i'}, i)$ denotes the amount of demand of $i'$ assigned to $i$ in $AS_{cap}$ and $\rhobar^1(j_i, i')$ denotes the number of clients co-located at $i$ and assigned to $i'$ in $AS_1$. 
Clients are assigned in three steps. In the first step, assign the clients co-located at a facility to itself. 
For small $i'$, additionally, we do the following  (type - $1$) re-assignments (see Fig. \ref{type1}): $(i)$ For small $i \ne i'$, assign $\rhobar^c(j_{i'}, i)$ clients co-located at $i$ to $i'$. Thus, $\rhobar^1(j_i, i') = \rhobar^c(j_{i'}, i)$ for small $i$. $(ii)$ For big $i$, assign $\rhobar^c(j_{i'}, i_1) + \rhobar^c(j_{i'}, i_2)$ clients co-located at $i$ to $i'$. Thus, $\rhobar^1(j_i, i') = \rhobar^c(j_{i'}, i_1) + \rhobar^c(j_{i'}, i_2)$ for big $i$. 
Claim~\ref{claim1} shows that these assignments are feasible and the capacities are violated only upto the extent to which they were violated by the tri-criteria solution $S^t$. 
\begin{figure}[t]
    \centering
    \includegraphics[width=8cm]{ 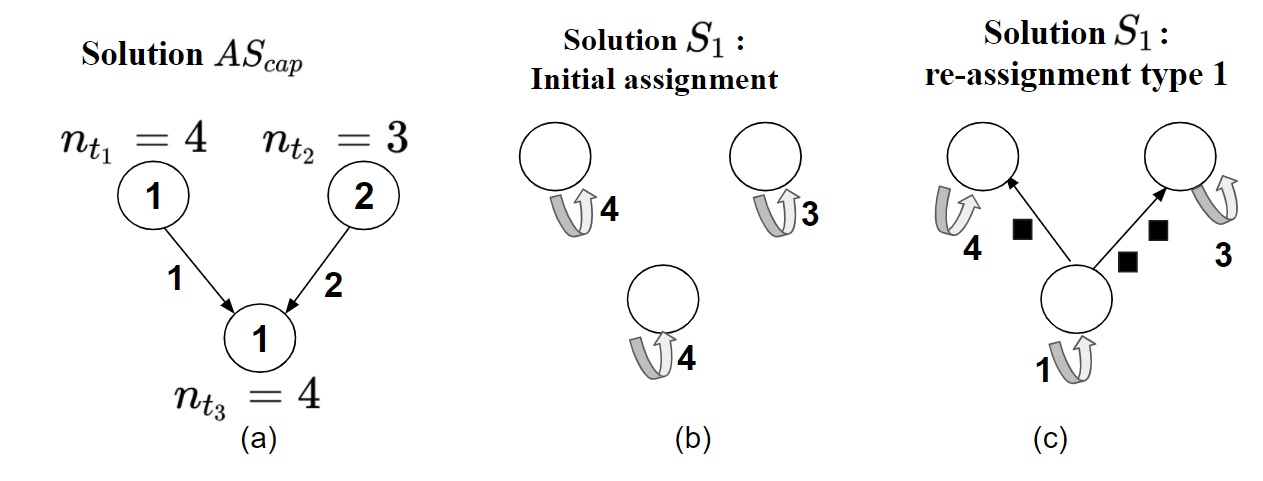}
    \caption{($a$) $\B =5$, $t_1, t_2, t_3$ have demands $1$, $2$ and $1$ unit each respectively. In solution $AS_{cap}$, $1$ unit of demand of $t_1$ and $2$ units of demand of $t_2$ are assigned to $t_3$. (b) Solution $S_1$: initially, $n_i$ clients are initially assigned to facility $i$. (c) Type $1$ reassignments: $1$ and $2$ clients of $t_3$ reassigned to $t_1$ and $t_2$ respectively. After this reassignment, $t_3$ has only $1$ client.}
    \label{type1}
\end{figure}

\begin{claim}
\label{claim1}
    $(i)$ $\sum_{i' \ne i} \rhobar^1(j_i, i') \le n_i$, $\forall$ $i \in \bar \facilityset^t$.
    $(ii)$ $\sum_{i} \rhobar^1(j_i, i') \le max\{\B, n_{i'}\} \le \beta \capacityU$, 
    $\forall$ $i' \in \bar \facilityset^t$.
 \end{claim}

\begin{proof}

    $(i)$ Since $\rhobar^c(j_{i}, i) = d_i$ therefore $\sum_{i' \ne i} \rhobar^1(j_i, i') = \sum_{i' \ne i} \rhobar^c(j_{i'}, i) \le u_i - d_i \le n_i$ in all the cases. $(ii)$ $\rhobar^1(j_{i'}, i') = n_{i'}$. For $ i \ne i'$, $\rhobar^1(j_i, i') = \rhobar^c(j_{i'}, i)$. Thus, $\sum_{i} \rhobar^1(j_i, i') = n_{i'} + \sum_{i \ne i'} \rhobar^c(j_{i'}, i) \le n_{i'} + d_{i'} \le max\{\B, n_{i'}\} \le \beta \capacityU$.
\end{proof}

Note that $AS_1$ so obtained may  still not satisfy the lower bound requirement. In fact, although each facility was assigned $n_i \ge \alpha \B$  clients initially, they may be serving less clients now after type $1$ re-assignments. For example, in Fig. \ref{type1}, $t_3$ had $4$ clients initially which was reduced to $1$ after type $1$ reassignments. Let $P \subseteq  \facilityset^t$ be the set of facilities each of which is serving at least $\B$ clients after type-$1$ reassignments, and  $\bar P =  \facilityset^t \setminus P$ be the set of remaining facilities. 
 We open all the facilities in $P$ and let them serve the clients assigned to them after type $1$ reassignments.
 
 \begin{observation}
    If a  small facility $i'$ was closed in $AS_{cap}$ then $i'$ is in $P$: $\rhobar^1(j_{i'}, i') +  \sum_{i \ne i'} \rhobar^1(j_i, i')=$ $n_{i'} +  \sum_{i \ne i'} \rhobar^c(j_{i'}, i) = n_{i'} + d_{i'} \ge \B$. Similarly for a big facility $i'$, if $i'_1$ was closed in $AS_{cap}$ then $i'$ is in $P$.

 \end{observation}

 Thus, a small facility is in $\bar P$ only if it was open in $AS_{cap}$ and a big facility $i$ is in $\bar P$ only if $i_1$ was open in $AS_{cap}$. 
 We now group these facilities so that each group serves at least $\B$ clients and open one facility in each group that serves all the clients in the group. For this we construct what we call as {\em facility trees}. We construct a graph $G$ with nodes corresponding to the facilities in $P \cup \bar P$. For $i \in \bar P$, let $\eta(i)$ be the nearest other facility to $i$ in $\facilityset^t$. Then, $G$ consists of edges $(i, \eta(i))$ with edge costs $ \dist{i}{\eta(i)}$. Each component of $G$ is a tree except possibly a double-edge cycle at the root. In this case, we say that we have a root, called {\em root-pair} $\rootpair$ consisting of a pair of facilities $i_{r_1}$and $ i_{r_2}$. Also, a facility $i$ from $P$, if present, must be at the root of a tree. Clearly edge costs are non-increasing as we go up the tree. 

Now we are ready to define our second type of re-assignments.
Let $x$ be a node in a tree $\mathcal T$. Let $children(x)$ denote the set of children of $x$ in $\mathcal T$. 
If $x$ is a root-pair $\rootpair$, then $children(x) = children(r_1) \cup children(r_2)$. 
Process the tree bottom-up (level by level) where processing of a node $x$ is explained in  Algorithm $1$, 
see Fig: (\ref{Process-X}). 
 While processing a node $x$, we first open (in lines $2-6$) and remove all its children with at least $\B$ clients; remaining children of $x$ are arranged and considered (left to right) in non-increasing order of distance from $x$.
 For any child $y \in children(x)$, let $right-sibling(y)$ denote the adjacent right sibling of $y$ in the arrangement; thus, $right-sibling(y_i) = y_{i+1}$.

\newpage

\input{Gen-AlgoSG.tex}
  \begin{figure}[t]
    \centering
    \includegraphics[width=12cm]{ 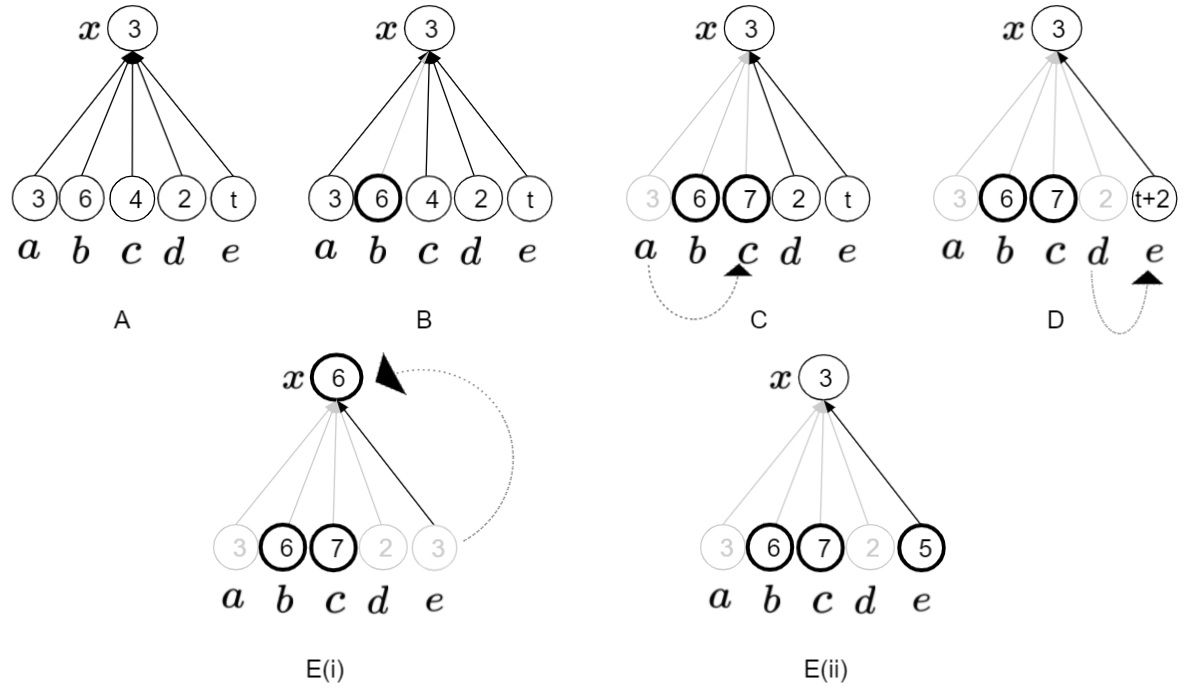}
     \caption{   Let $\B=5$, numbers inside a node represents the number of clients at the node. A dashed arrow from a node $u$ to node $v$ represents the movement of clients from $u$ to $v$; (A): A node $x$ along with its children in decreasing order of their distances from $x$; (B): In steps $1$ to $7$, we open child $b$ and delete the edge $(b, x)$; (C): In steps $18 -20$, clients from $a$ are assigned to $c$ accumulating a total of $7$ clients at $c$. In the next iteration of the 'for' loop, in lines $14-16$, $c$ gets opened and the edge ($c, x$) is deleted; (D): In the next iteration, clients of $d$ are assigned to $e$; (E(i)): $t=1$.  In the next iteration, clients of $e$ are assigned to $x$ in lines $27 - 28$; (E(ii)): $t=3$. $e$ gets opened in lines $23 - 25$;} 
    \label{Process-X}
\end{figure}

There are two possibilities at the root: either we have a facility $i$ from $P$ or we have a root-pair $\rootpair$.
In the first case, we are done as $i$ is already open and has at least $\B$ clients. See Fig: (\ref{tree12}) and (\ref{tree2}). To handle the second case, we need to do a little more work. So, we define our assignments of third type as follows: 
($i$) If the total number of clients collected at  the root-pair node 
is at least $\B$ and at most $2\B$ then open any one of the two facilities in the root-pair and assign all the clients to it.\footnote{ we can open the facility with more number of clients and save $2$(to be seen again) factor in the connection cost.} See Fig: (\ref{tree3}-(A)).
($ii$) If the total number of clients collected at  the root-pair node is more than $2\B$  then open both the facilities at the root node and distribute the clients  so that each one of them gets at least $\B$ clients. See Fig: (\ref{tree3}-(B)).
($iii$) If the total number of clients collected at  the root-pair node  is less than $\B$, then let $i$ be the node in $P$ nearest to the root-pair i.e. $i =  argmin_{i' \in P} \min \{ \dist{i'}{r_1}, \dist{i'}{r_2}\}$, then $i$ is already open and has at least $\B$ clients. Send the clients collected at the root-pair to $i$. See Fig: (\ref{tree4}).

\begin{figure}[t]
	\begin{tabular}{ccc}
		\includegraphics[width=9cm,scale=0.8]{ 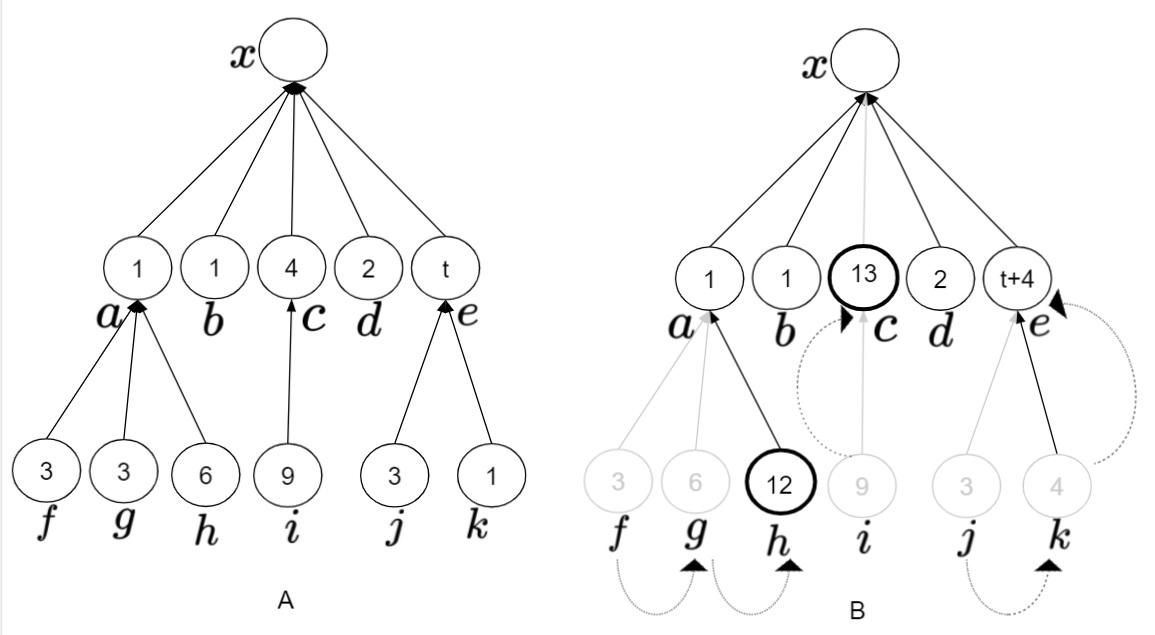}
		& 
		\includegraphics[width=4cm,scale=0.8]{ 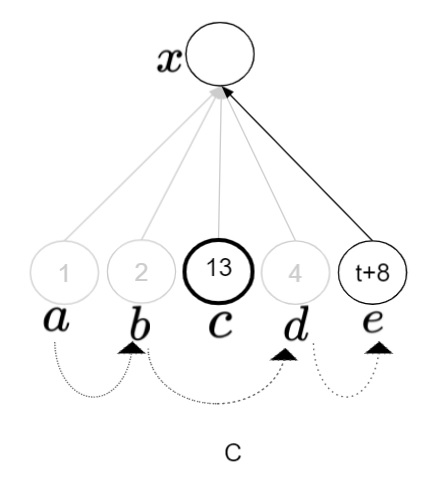}
		\\
	\end{tabular}
	\caption{$\B=12$. (A) A tree $\mathcal{T}$;  Numbers inside each node represents the number of clients at those nodes.
	The tree does not change after the calls: Process($f$), Process($g$), Process($h$), Process($i$), Process($j$) and Process($k$). In each case, it returns at line no. $9$;  (B) Tree after Process($a$), Process($b$), Process($c$), Process($d$) and Process($e$); (C) Tree after line $22$ of Process($x$).
	}
	
	\label{tree12}
\end{figure}


 \begin{figure}[t]
    \centering
    \includegraphics[width=4cm]{ 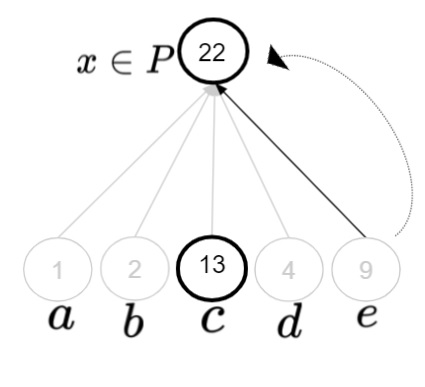}
    \caption{ $\B=12,~t=1$. Before Process($x$): $n_x = 13$. Tree after Process($x$) when root node $x \in P$.  }
    \label{tree2}
\end{figure}

 \begin{figure}[t]
    \centering
    \includegraphics[width=10cm]{ 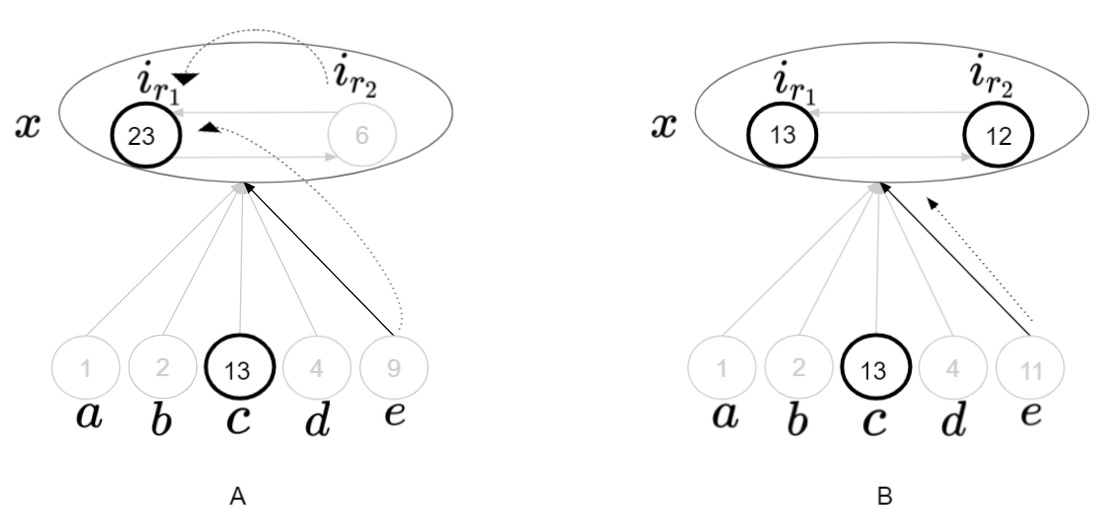}
    \caption{    $\B=12$, $x = <i_{r_1}, i_{r_2}>$ is a root-pair node with $n_{i_{r_1}} = 8,~n_{i_{r_2}} = 6$. (A) $t =1$. Before Process($x$): $\B \le n_{i_{r_1}} + n_{i_{r_2}} + 9 \le 2\B$. Tree after Process($x$): any one of $i_{r_1}$ and $i_{r_2}$, say $i_{r_1}$, is opened and all the clients are assigned to it; (B)$t =3$. Before Process($x$): $2\B \le n_{i_{r_1}} + n_{i_{r_2}} + 11 \le 3\B$. Tree after Process($x$): $i_{r_1}$ and $i_{r_2}$ are both opened and clients are distributed among them.}
    
    \label{tree3}
\end{figure}

 \begin{figure}[t]
    \centering
    \includegraphics[width=5cm]{ 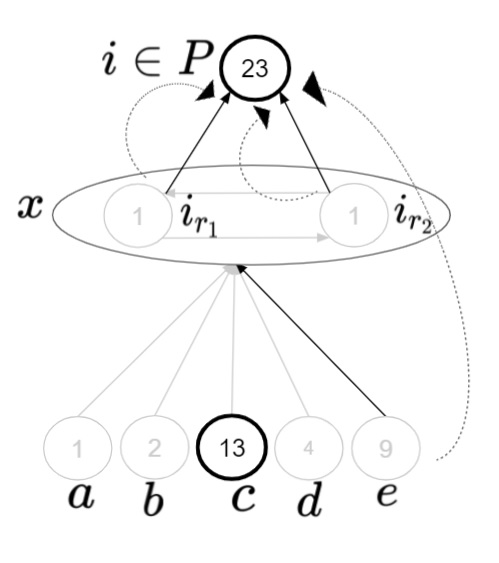}
    \caption{$\B=12, t =1$. $x = <i_{r_1}, i_{r_2}>$ is a root-pair node with $ n_{i_{r_1}} + n_{i_{r_2}} + 9 < \B$. Tree after Process($x$).}
    
    \label{tree4}
\end{figure}

Clearly the opened facilities satisfy the lower bounds. Next, we bound the violation in the upper bound.
  
\begin{lemma}

The number of clients collected in the second type of re-assignments, at any non-root node $x$ is at most $2\B$.
 \end{lemma}
 
 \begin{proof} 
 Note that a node $x$ receives clients either from one of its children (at line $27$) when Process($x$) is invoked  or from one of its siblings (at line $18$) when Process($\eta(x)$) is invoked. In either case, it has $< \B$ clients of its own and it receives $<\B$ clients from its child/sibling.  Hence, the claim follows.
 \end{proof}
  
  If the number of clients collected at $i_c$ is $\ge \B$, it is opened (at line $24$) and removed from further processing when Process($\eta(i_c)$) is invoked. 
 Let $\pi_i$ be the number of clients assigned to the facility $i$ after first type of re-assignments. Then as shown earlier, $\pi_i \le \beta \capacityU$.  
 We have the following cases: ($i$) the root node $i$ is in $P$ and it gets additional  $< \B$ clients from $i_c$ making a total of at most $\pi_i + \B \le \beta \capacityU + \B \le (\beta + 1) \capacityU$ clients at $i$.
 ($ii$) the total number of clients collected at  the root-pair node 
is at least $\B$ and at most $2\B $.
 The number of clients the opened facility gets is at most $2\B \le 2\capacityU \le (\beta + 1) \capacityU$.
 ($iii$)  the total number of clients collected at the root-pair node  is more than $2\B$.  Note that the total number of clients collected at these facilities is at most $3\B$ and hence 
 ensuring that each of them gets at least $\B$ clients also ensures that none of them gets more than $2\B$ clients. Hence none of them gets more than $2\capacityU \le (\beta + 1) \capacityU$ clients.
 ($iv$) the total number of clients collected at the root-pair node is less than $\B$. 
 Thus, the number of clients the opened facility in $P$ gets is at most $\beta \capacityU + \B \le (\beta + 1) \capacityU$. 

 We next bound the connection cost.  
 We will need the following lemma to bound the connection cost.

\begin{lemma}
 For a node $y$ in a tree,
  $\dist{y}{right-sibling(y)} \le 3 \dist{y}{\eta(y)} = 3l(y)$. 
 \end{lemma}
 \begin{proof}
 If $\eta(\cone) = \eta(right-sibling(\cone))$  (i.e. $\cone$ and $right-sibling(\cone)$ are children of the same node), then it is easy to see that $\dist{\cone}{right-sibling(\cone)} \leq 2 \dist{\cone}{\eta(\cone)}$. Otherwise (which can be the case when root is a root-pair, see Figure~(\ref{3factorBound})), $\dist{\cone}{right-sibling(\cone)} \leq \dist{\cone}{\eta(\cone)} + \dist{\eta(\cone)}{\eta(right-sibling(\cone))} + \dist{\eta(right-sibling(\cone))}{right-sibling(\cone)} \le 3 \dist{\cone}{\eta(\cone)}$.
 \end{proof}
 
  \begin{figure}[t]
    \centering
    \includegraphics[width=7cm]{ 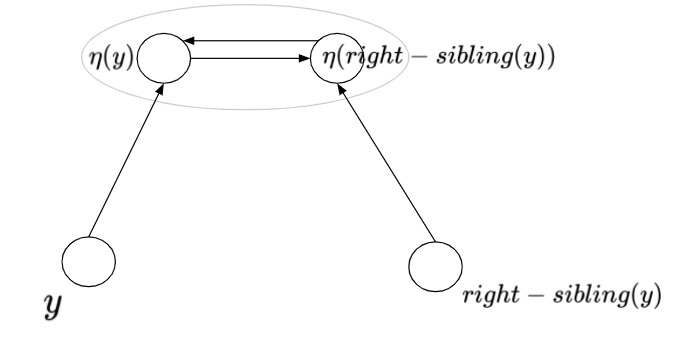}
    \caption{$\dist{\cone}{right-sibling(\cone)}  \le 3 \dist{\cone}{\eta(\cone)}$}
    \label{3factorBound}
\end{figure}

 The cost of first type of re-assignments is at most the connection cost of $AS_{cap}$.
 To bound the connection cost of second type, observe that we never send more than $\B$ clients on any edge. 
 Thus, a facility $i$ that sends its (at most $\B$ clients) to its parent or to its right sibling incurs a cost of at most $3\B l(i)$. Since the facility $i$ was opened in $AS_{cap}$, it pays facility opening cost of $\delta l(i) min \{n_i, \B\} \ge \delta l(i) \alpha \B$ in $AS_{cap}$. Summing over all $i \in \bar P$ for the second type of re-assignments, we see that the total connection cost in these re-assignments is bounded by  $3 FCost_{I_{cap}}(AS_{cap})/\alpha \delta$ 
 where $FCost_{I_{cap}}(.)$ is the facility opening cost of a solution to the $I_{cap}$ instance.

 For the third type of re-assignments, we bound the cost as follows:
 the cost of sending ($< \B$) clients by a facility $i$ in the root pair $\rootpair$ to the facility in $P$, nearest to the root-pair, is not bounded by $\B l(i)$ (where $l(i) = l(i_{r_1}) = l(i_{r_2})$). 
 Note that the total number of clients, initially at $i_{r_1}$ and $i_{r_2}$ was $\ge 2 \alpha \B$. However, during re-assignments of type $1$, some of them got reassigned to other facilities. Note that these (other) facilities must have been closed in $AS_{cap}$ and hence do not belong to $\bar P$. Hence, each of these reassignments correspond to an assignment in $AS_{cap}$ whose cost was at least $l(i)$ and a total of at least $(2\alpha  - 1)\B l(i)$. Hence the cost of type $3$ re-assignments is bounded by $3 CCost_{I_{cap}}(AS_{cap})/ (2\alpha  - 1)$ where $CCost_{I_{cap}}(.)$ is the connection cost of a solution to the $I_{cap}$ instance. Hence, the total cost of solution $AS_1$  is $max\{1+\frac{3}{(2\alpha-1)},\frac{3}{\alpha\delta}\}Cost_{I_{cap}}(AS_{cap}) = max\{\frac{2(\alpha+1)}{(2\alpha-1)},\frac{3}{\alpha\delta}\} Cost_{I_{cap}}(AS_{cap}) = \frac{2(\alpha+1)}{(2\alpha-1)} Cost_{I_{cap}}(AS_{cap}) $ for $\delta= \frac{3(2\alpha-1)}{2\alpha(\alpha+1)} $.

Thus, the total cost of the solution is,\\ $ Cost_I(S) \le  Cost_I(S^t) + Cost_{I_1}(AS_1)$\\
   $\leq O(1) Cost_I(O) + \frac{2(\alpha+1)}{(2\alpha-1)} Cost_{I_{cap}}(AS_{cap}) $\\
    $\leq O(1) Cost_I(O) +  \frac{2(\alpha+1)(5+\epsilon)}{(2\alpha-1)} Cost_{I_{cap}}(O_{cap})$\\
   $\leq O(1) Cost_I(O) +  \frac{12(\alpha+1)}{(2\alpha-1)} Cost_{I_{cap}}(O_{cap})$\\
   $ \leq O(1) Cost_I(O) + \frac{12(\alpha+1)(1 + 2\delta)}{(2\alpha-1)}  Cost_{I_2}(O_2) $\\
   $\leq O(1) Cost_I(O) + \frac{12(\alpha^2 + 7 \alpha -3)}{\alpha (2\alpha-1)}  Cost_{I_2}(O_2)$\\
   $= O(1) Cost_I(O) +  \frac{24(\alpha^2 + 7 \alpha -3)}{\alpha (2\alpha-1)}  Cost_{I_1}(O_1)$\\
    $= O(1) Cost_I(O) +  \frac{24(\alpha^2 + 7 \alpha -3)}{\alpha (2\alpha-1)} (Cost_I(S^t) + Cost_I(O))$\\
   $= O(1) Cost_I(O) +  \frac{24(\alpha^2 + 7 \alpha -3)}{\alpha (2\alpha-1)} (O(1) Cost_I(O))$\\
   $=O(1) Cost_I(O)$
   

 

%% file: Gen-AlgoSG.tex
\begin{algorithm}[H] 
    \footnotesize
    \setcounter{AlgoLine}{0}
	\SetAlgoLined
    \DontPrintSemicolon  
    \SetKwInOut{Input}{Input}
	\SetKwInOut{Output}{Output}
    \Input{$x ( x$ can be a root-pair node)}
    
	\For {$ y \in children(x)$} {
	 \If{$n_{y} \geq \B$}  {
	    Open facility $y$ \;
		 
		 Remove edge ($y, \eta(y)$) \tcp*{Remove the connection of $y$ from its parent}
		 
		 Delete $y$ from $children(x)$;
	 }
	 }
	\If {$children(x) =  \phi$} {return;}
	
	Arrange $children(x)$ in the sequence $<y_1, \ldots y_k>$ such that $\dist{y_{i}}{\eta(y_{i})} \geq \dist{y_{i + 1}}{\eta(y_{i + 1})} ~\forall ~i = 1 \ldots k - 1$\;
	
	 
		 
		 \For {$i = 1$ to $k-1$} {  
		 
		     \eIf{$n_{y_i} \geq \B$}  {
		 
		        Open facility $y_i$\;
		
		     Remove edge ($y_i, \eta(y_i)$) \tcp*{Remove the connection of $y_i$ from its parent}
		 
		     Delete $y_i$ from $children(x)$;
		    }
		    {
		     $n_{y_{i+1}} = n_{y_{i + 1}} + n_{y_{i}}$	 \tcp*{Send the clients of $y_i$ to $y_{i+1}$}
		 
		      Remove edge ($y_{i}, \eta(y_{i})$) \tcp*{Remove the connection of $y_i$ from its parent}
		 
		      Delete $y_{i}$ from $children(x)$\;
		    }
		 }
		 
		 \eIf {$n_{y_k} \geq \B$}  {
		 
		 Open facility $y_k$ \;
		 
		 
		 Delete $y_k$ from $children(x)$\;
		 }
		 {
		 $n_{\eta(y_k)} = n_{\eta(y_k)} + n_{y_{k}}$ \tcp*{Send the clients of $y_k$ to $\eta(y_k)$}
		 
		 
		 Delete $y_k$ from $children(x)$\;
		 }
		 
		 Return\;
	
	\label{TreeProcessing}
	\caption{Process($x$)}
\end{algorithm}

%% file: Conclusion.tex
\section{Conclusion and Future Work}
In this paper, we presented 
the first (constant) approximation algorithm for facility location problem with uniform lower and upper bounds without violating the lower bounds.
Upper bounds are violated by ($5/2$)-factor. 
Violation in the upper bound is less when the gap between the lower and the upper bound is large. For example, if $\B \le \capacityU/2$ then 
the upper bound violation is at most $\beta + \frac{1}{2} \le 2$.


In future, if one can obtain a tri-criteria solution (with $\alpha > 1/2$) when one of the bounds is uniform and the other is non-uniform, then it can be simply plugged into our technique to obtain similar result for the problem. When the upper bounds are non-uniform and 
$\tau = \frac{1}{\max_i \{\capacityU_i/\B \}}$ then $ \B \le \tau \capacityU_i$ for all $i$ and $\tau \le 1$. Also, violation in the upper bounds, then, is $\le \beta + \tau \le \beta + 1$. When the lower bounds are non-uniform and $\tau = \frac{1}{\max_i \{\capacityU/\B_i \}}$ then $ \B_i \le \tau \capacityU$ for all $i$ and $\tau \le 1$. As before, violation in the upper bounds is $\le \beta + \tau \le \beta + 1$.
Note that, for this case, we can not use the tri-criteria solution of Friggstad \etal~\cite{friggstad_et_al_LBUFL} as $\alpha < 1/2$ in their solution. 


